\documentclass[noinfoline, 11pt]{imsart} 

\usepackage[utf8]{inputenc} 

\usepackage[left = 1.1in, right = 1.1in, top = 1.5in, bottom = 1.5in]{geometry} 

\usepackage{amsthm, amsfonts, amssymb, amsmath,enumitem, natbib, bbm,mathabx}
\usepackage{mathtools}
\mathtoolsset{showonlyrefs}
\newtheorem{thm}{Theorem}[section]
\newtheorem{lem}{Lemma}[section]

\newtheorem{prop}{Proposition}[section]
\newcounter{myalgctr}
\newenvironment{rem}{
   \vskip1mm\indent
   \refstepcounter{myalgctr}
   \textbf{Remark \themyalgctr}
   }{\hfill$\diamond$\par}  
\numberwithin{myalgctr}{section}
\providecommand{\norm}[1]{\left\lVert#1\right\rVert}
\DeclareMathOperator*{\argmin}{arg\,min}
\DeclareMathOperator*{\argmax}{arg\,max}

\usepackage[colorlinks=true, pdfstartview=FitV,
linkcolor=blue, citecolor=blue, urlcolor=blue]{hyperref}
\makeatletter
\def\namedlabel#1#2{\begingroup
    #2%
    \def\@currentlabel{#2}%
    \phantomsection\label{#1}\endgroup
}
\makeatother
\newcommand{\vertiii}[1]{{\left\vert\kern-0.25ex\left\vert\kern-0.25ex\left\vert #1
    \right\vert\kern-0.25ex\right\vert\kern-0.25ex\right\vert}}

\newcommand{\Leb}{\mbox{\textbf{Leb}}}
\newcommand{\reals}{\mathbb{R}}

\begin{document}
 \begin{frontmatter}
\title{Valid Post-selection Inference in Assumption-lean Linear Regression}
\runtitle{PoSI for Linear Regression}
\begin{aug}
  \author{\fnms{Arun K.}  \snm{Kuchibhotla}\ead[label=e1]{arunku@wharton.upenn.edu}},
  \author{\fnms{Lawrence D.} \snm{Brown}},
  \author{\fnms{Andreas} \snm{Buja}},
  \author{\fnms{Edward I.} \snm{George}},
  \and
  \author{\fnms{Linda H.} \snm{Zhao}}

  \runauthor{Kuchibhotla et al.}

  \affiliation{University of Pennsylvania}

  \address{University of Pennsylvania\\ \printead{e1}}

\end{aug}
\begin{abstract}
Construction of valid statistical inference for estimators based on data-driven selection has received a lot of attention in the recent times. \cite{Berk13} is possibly the first work to provide valid inference for Gaussian homoscedastic linear regression with fixed covariates under arbitrary covariate/variable selection. The setting is unrealistic and is extended by \cite{Bac16} by relaxing the distributional assumptions. A major drawback of the aforementioned works is that the construction of valid confidence regions is computationally intensive. In this paper, we first prove that post-selection inference is equivalent to simultaneous inference and then construct valid post-selection confidence regions which are computationally simple. Our construction is based on deterministic inequalities and apply to independent as well as dependent random variables without the requirement of correct distributional assumptions. Finally, we compare the volume of our confidence regions with the existing ones and show that under non-stochastic covariates, our regions are much smaller. 
\end{abstract}
\end{frontmatter}


\section{Introduction and Motivation}

\subsection{Motivation of the Problem}

In recent times, there has been a crisis in the sciences because too many research results are found to lack replicability and reproducibility. Some of this crisis has been attributed to a failure of statistical methods to account for data-dependent exploration and modeling that precedes statistical inference.  Data-dependent actions such as selection of subsets of cases, of covariates, of responses, of transformations and of model types has been aptly named ``researcher degrees of freedom'' \citep{simmons2011false}, and these may well be a significant contributing factor in the current crisis.  Classical statistics does not account for them because it is built on a framework where all modeling decisions are to be made {\em independently of the data on which inference is to be based}. But if the data are in fact used to this end prior to statistical inference, then such inference loses its justifications and the ensuing validity conferred on it by classical theories. It is therefore critical that the theory of statistical inference be brought up to date to account for data-driven modeling. Updating the theory that justifies statistical inferences usually requires modifying the procedures of inference such as hypothesis tests and confidence intervals. As a consequence, the new procedures may lose some power relative to the previously stipulated but illusionary power derived from classical theories. This, however, is a necessary price to be paid for better justification of statistical inference in the context of the pre-inferential liberties taken in today's data-analytic practice. While updating of statistical theories and inference procedures will not solve all problems underlying the current crisis, it is a necessary step as it may help mitigate at least some aspects of the crisis.  In what follows we refer to all data-analytic decisions that are made using the data prior to inference as ``data-driven modeling''.

A second issue with theories of classical statistical inference is that many of them rely on the assumption that the data have been correctly modeled in a probabilistic sense. This means the theories tend to assume that the probability model used for the data correctly captures the observable features of the data generating process. Justifications of statistical inferences derived from such theories are therefore {\em in}valid if the model is incorrect or (using the technical term) ``misspecified''. With the proliferation of data-analytic approaches in science and business, it is becoming ever more unrealistic to assume that all statistical models are correctly specified and inferences are made only after carefully vetting the model for correct specification, for example, using model diagnostics.  Such vetting may never have been realistic in the first place, and it should also be said that pre-inferential diagnostics should be counted among ``researcher degrees of freedom'' as they may result in data-driven modeling decisions. It is therefore a mandate of realism to use so-called ``model-robust'' methods of statistical inference, and for statistical theory to provide their justifications. In matters of misspecification the situation is somewhat less dire than data-driven modeling as there exists a rich literature on the study of inference when models are misspecified.
We will naturally draw on extant proposals for misspecification-robust or (using the technical term) ``model-robust'' inference and adapt them to our purposes.

To summarize, there exist at least two ways in which statistical inferences with justifications from classical mathematical statistics only can be invalidated, namely,
\begin{enumerate}[label = \bfseries(P\arabic*)] \itemsep 0em
\item data-driven modeling prior to statistical inference, and\label{prob1}
\item model misspecification.\label{prob2}
\end{enumerate}
In light of the replicability and reproducibility crisis in the sciences, it is of considerable interest, even urgency, to develop methods of statistical inference and associated theoretical justifications that account for both \ref{prob1} and \ref{prob2}. Even though these problems are manifest in almost all statistical procedures used in practice, it is no simple task to provide methods of valid statistical inference that address these problems in greater generality. For this reason the present article puts forth specifically a method of valid inference for the case that the fitting procedure is ordinary least squares (OLS) linear regression. Here there exists a literature that documents the drastic effects of ignoring \ref{prob1} and \ref{prob2}; see, for example, \cite{Bueh63}, \cite{Olsh73}, \cite{REN80}, and~\cite{Freed83}. We will address one particular form of problem \ref{prob1}, namely, data-driven selection of regressor variables, and we will deal with several forms of problem~\ref{prob2}.

Some of the earliest work that studies estimators under data-dependent modeling \ref{prob1} include \cite{Hjort03} and \cite{Claeskens07}. Although these articles deal with a general class of statistical procedures, a major limitation, in view of the current article, is that the data-dependent modeling is restricted to a very narrow class of principled variable selection methods such as AIC or some other information criterion. The fact is, however, that few data analysts will confine themselves to a strict protocol of data-driven modeling.  To address broader aspects of ``researcher degrees of freedom'' there have more recently emerged proposals that provide validity of statistical inference in the case of arbitrary data-driven selection of regressor variables. The first such proposal was by \cite{Berk13} who solve the problem allowing misspecified response means but retaining the classical assumptions of homoskedastic and normally distributed errors. We refer to \cite{Berk13} for many other prior works related to problem \ref{prob1} where data-driven modeling consists of selection of regressor variables. A more recent article that expands on \cite{Berk13} is by \cite{Bac16}. An alternative approach is by  \cite{Lee16}, \cite{Tibs16}, \cite{Tian16} (for example). Similar to \cite{Hjort03}, these proposals do not insure validity of inference against arbitrary regressor selection but against {\em specific selection methods} such as the lasso or stepwise forward selection. This type of post-selection inference is conditional on the selected model and dependent on distributional assumptions, thereby not addressing problem~\ref{prob2}.

The present article is close in spirit to \cite{Berk13} and \cite{Bac16} and lends their approach a considerable degree of generality by covering both fixed regressors (as in these references) and (newly) random regressors. \cite{Bac16} is the only work we know of that provides valid statistical inference under arbitrary data-dependent regressor selection and general misspecification of the regression models. Their framework assumes a situation where the set of sub-models is finite and of fixed cardinality independent of the sample size. Their method of statistical inference is NP-hard, hence requires computational heuristics. To overcome these limitations we propose here a simplified procedure with the following properties: (1)~it is comparatively computationally efficient with at most polynomial complexity in the total number of covariates, and (2)~it allows the set of sub-models to grow almost exponentially as a function of the sample size. Thus the procedure is also in the spirit high-dimensional statistics where the total number of covariates is allowed to be much larger than the sample size.


\subsection{Overview}

In what follows, the term ``model-selection'' will always mean arbitrary data-driven selection of regressor variables, which is the only aspect of problem \ref{prob1} that will be addressed in this article. Furthermore, the only fitting method considered here is OLS linear regression; this limitation is for expository purposes, and results for more general types of regressions will be given elsewhere. Problem \ref{prob2} will be addressed by the complete absence of modeling assumptions. In particular, it will {\em not} be assumed that the response means behave linearly in the regressors, and equally it will {\em not} be assumed that the errors are homoskedastic and normally distributed.
The goal is to provide confidence regions for linear regression coefficients obtained after model-selection. In the process, we will prove simple but powerful results about linear regression that lend themselves to proving the validity of confidence regions. The main contributions of the current paper are as follows:
\begin{enumerate}
\item We treat OLS linear regression as a fitting method for linear equations while treating the associated Gaussian linear model merely as a working model that is permitted to be misspecified. We consider the case where the observations are the random vectors comprised of a response variable and one or more regressor variables/covariates, allowing the latter to be random rather than fixed. Note that fixed covariates are assumed in the settings of \cite{Berk13} and \cite{Bac16}. Random covariates require us to interpret and understand what is being estimated more carefully. See \cite{Buja14} for an explanation why under misspecification the treatment of random covariates as fixed is not justified.
\item Following \cite{Berk13} and \cite{Bac16} we decouple the inference problem from model selection, meaning that the inferences proposed here are valid no matter how the model selection was done. This feature has pluses and minuses. On the plus side, inferences will be valid even in the presence of ad-hoc and informal selection decisions made by the data analyst, including, for example, visual diagnostics based on residual plots. On the minus side, decoupling implies that inferences cannot take into account any properties of the model selection procedure when in fact only one such procedure was used. A strong argument by  \cite{Berk13} and \cite{Bac16} in favor of decoupling, however, is that in reality data analysts will rarely limit themselves to one and only one formal selection method if it produces unsatisfactory results on the data at hand. Therefore, in order to truly contribute to solving the crisis in the sciences, unreported informal selection should be assumed and accounted for. Decoupling of model selection and inference has a further benefit: It solves the circularity problem by permitting selection to start over and over as often as the data analyst pleases; inferences in all selected models will be valid, whether they are found satisfactory or unsatisfactory for whatever reasons.
\item Our theory provides validity of post-selection inferences even when model selection is applied to a very large number of covariates --- almost exponential in the sample size. Thus the theory is in the spirit of contemporary high-dimensional statistics which is interested in problems where the number of variables is larger than the sample size. Of course we require model selection to produce models of size smaller than the sample size in order to avoid trivial collinearity when the number of covariates exceeds the sample size.
\item We mostly focus on one simple strategy for valid post-selection inference that has the advantage of great simplicity, both in theory and in computation --- its computational cost being proportional to the number $p$ of covariates. This is surprising as the computational complexity of \cite{Berk13} is exponential in $p$ because it requires searching all covariates in all possible submodels. The drawback of the strategy is that its confidence regions are not aligned with the coordinate axes in covariate space, hence do not immediately provide confidence intervals for the slope parameters of the form ``estimate $\pm$ half-width''.  
\item Most of the present results are based on deterministic inequalities that allow for valid post-selection inference even when the random vectors involved are structurally dependent. This approach may not produce best possible rates in some contexts, but the resulting inferences will be more robust to independence assumptions.
\end{enumerate}
As a caveat, it should be stated that we do not address the question of when linear regression is appropriate in a given data analytic situation when misspecification is present. We consider it a reality that many if not most linear regressions are fitted in the presence of various degrees of misspecification, and reporting results for interpretation should be accompanied by statistical inference just the same. Our goal is therefore limited to providing asymptotic justification of inference in the presence of misspecification {\em and} after data-driven model selection.

The remainder of the paper is organized as follows. Section \ref{sec:Notation_Problem_Formuation} provides the necessary notation for a rigorous formulation of the problem of valid post-selection inference. In Section \ref{sec:Simultaneous}, the problem of post-selection inference is shown to be equivalent to a problem of simultaneous inference. In Section \ref{sec:FirstStrat} we present the first strategy for valid post-selection inference along with its main features. Section \ref{sec:Comp} describes an implementation method based on the multiplier bootstrap. Section \ref{sec:Generalization} provides a simple generalization to linear regression-type problems. Section \ref{sec:HighDim} points out an interesting connection between the post-selection confidence regions proposed here and the estimators proposed in the high-dimensional linear regression literature. In Section \ref{sec:ProsAndCons}, we discuss various advantages and disadvantages of the approach presented in this paper. The final Section \ref{sec:Conclusions} summarizes the results.

Many of the proofs are deferred to Appendices \ref{app:ProofLemmaL1}, \ref{app:LebesgueMeasure} and \ref{app:LassoRegions}. Most of the discussion in the paper is based on the assumption of independent random vectors, although comments about applicability to dependent random vectors are given in appropriate places. Appendix \ref{app:HDCLT} provides theoretical background about a high-dimensional central limit theorem and the consistency of multiplier bootstrap. These results are required for computation of joint quantiles for the proposed confidence regions. Appendix \ref{app:Dependent} describes the functional dependence setting where the computation of required quantiles is not much different from that of the independence setting.


\section{Notation and Problem Formulation}\label{sec:Notation_Problem_Formuation}

\subsection{Notation related to Vectors, Matrices and Norms}  \label{sec:Notation}

For any vector $v\in\mathbb{R}^q$ and $1\le j\le q$, $v(j)$ denotes the $j$-th coordinate of $v$. For any non-empty subset $M \subseteq \{1,2,\ldots,q\}$, $v(M)$ denotes the sub-vector of $v$ with indices in $M$. For instance, if $M = \{2, 4\}$ and $q \ge 4$, then $v(M) = (v(2), v(4))$. If $M = \{j\}$ is a singleton then $v(j)$ is used instead of $v(\{j\})$.
Therefore, $ v(M) \in \mathbb{R}^{|M|}$ where $|M|$ denotes the cardinality of~$M$.

For any symmetric matrix $A\in\mathbb{R}^{q\times q}$ and $M \subseteq \{1,2,\ldots,q\}$, let $A(M)$ denote the sub-matrix of $A$ with indices in $M\times M$ and for $1\le j, k\le q$, let $A(j,k)$ denote the value at the $j$-th row and $k$-th column of $A$.

Define the $r$-norm of a vector $v\in\mathbb{R}^q$ for $1\le r\le \infty$ as usual~by
\[
\norm{v}_r := \bigg(\sum_{j=1}^q |v(j)|^r\bigg)^{1/r},\quad\mbox{for}\quad 1\le r < \infty,\quad\mbox{and}\quad \norm{v}_{\infty} := \max_{1\le j\le q}|v(j)|.
\]
Let $\norm{v}_0$ denote the number of non-zero entries in $v$ (note this is not a norm). For any symmetric matrix $A$, let $\lambda_{\min}(A)$ denote the minimum eigenvalue of $A$. Also, let the elementwise maximum and the operator norm be defined, respectively, as
\[
\norm{A}_{\infty} := \max_{1\le j, k\le q}|A(j,k)|,\quad\mbox{and}\quad \norm{A}_{op} := \sup_{\norm{\delta}_2 \le 1}\norm{A\delta}_2.
\]
The following inequalities will be used throughout without special mention:
\begin{equation}\label{eq:MatrixVectorIneq}
\norm{v}_1 \le \norm{v}_0^{1/2}\norm{v}_2,\quad \norm{Av}_{\infty} \le \norm{A}_{\infty}\norm{v}_1,\quad\mbox{and}\quad |u^{\top}Av| \le \norm{A}_{\infty}\norm{u}_1\norm{v}_1,
\end{equation}
where $A\in\mathbb{R}^{q\times q}$ and $u, v\in\mathbb{R}^q$.


\subsection{Notation Related to Regression Data and OLS}

Let $(X_i^{\top}, Y_i)^{\top}\in\mathbb{R}^p\times\mathbb{R} ~ (1\le i\le n)$ represent a sample of $n$ observations. The covariate vectors $X_i \in \mathbb{R}^p$ are column vectors. It is common to include an intercept term when fitting the linear regression. To avoid extra notation, we assume that all covariates under consideration are included in the vectors $X_i$, so the data analyst may take the first coordinate of $X_i$ to be 1. In case that the number $p$ of covariates varies with $n$, this should be interpreted as a triangular array. Throughout, the term ``model'' is used to refer to the subset of covariates present in the regression and there will be \underline{\emph{no}} assumption that any linear model is true for any choice of covariates.

In order describe ``models'' in the sense of subsets of covariates, we use index sets $M \subseteq \{1,2,\ldots,p\}$ as in the previous subsection and write $X_i(M)$ for the covariate vectors in the submodel $M$. For any $1\le k\le p$, define the set of all non-empty models of size no larger than $k$ by
\begin{equation}\label{eq:kSparse}
\mathcal{M}_p(k) := \{M:\, M\subseteq\{1,2,\ldots, p\},\,\, 1\le |M| \le k\},
\end{equation}
so that $\mathcal{M}_p(p)$ is the power set of $\{1,2,\ldots,p\}$ excluding the empty set.

To proceed further, we assume that the observations are independent but possibly non-identically distributed. Note that this assumption includes as special cases (i)~the setting of independent and identically distributed observations and (ii)~the setting of fixed (non-random) covariates (by defining the distribution of $X_i$ to be a point mass at the observed $X_i$). Our setting is more general than either (i) or (ii) in that some of the covariates are allowed to be fixed while others are random.

For any $M \subseteq \{1,2,\ldots,p\}$, define the ordinary least squares empirical risk (or objective) function as
\begin{equation}\label{eq:EmpObj}
\hat{R}_n(\theta; M) := \frac{1}{n}\sum_{i=1}^n \left\{Y_i - X_i^{\top}(M)\theta\right\}^2,\quad\mbox{for}\quad \theta\in\mathbb{R}^{|M|}.
\end{equation}
Using this, define the expected risk (or objective) function as
\begin{equation}\label{eq:PopObj}
R_n(\theta; M) := \frac{1}{n}\sum_{i=1}^n \mathbb{E}\left[\left\{Y_i - X_i^{\top}(M)\theta\right\}^2\right],\quad\mbox{for}\quad \theta\in\mathbb{R}^{|M|}.
\end{equation}
(The notations $\mathbb{E}$ and $\mathbb{P}$ are used to denote expectation and probability computed with respect to all the randomness involved.)
Define the least squares estimator and the corresponding target for model $M$ as
\begin{equation}\label{eq:LeastSquares}
\hat{\beta}_{n,M} := \argmin_{\theta\in\mathbb{R}^{|M|}} \hat{R}_n(\theta; M),\quad\mbox{and}\quad \beta_{n,M} := \argmin_{\theta\in\mathbb{R}^{|M|}}R_n(\theta; M),
\end{equation}
for all $M\subseteq\{1,2,\ldots,p\}$, hence $\hat{\beta}_{n,M}, \, \beta_{n,M} \!\in\! \mathbb{R}^{|M|}$. Note, however, the following: Suppose $M = \{1, 2\}$ and $M' = \{1\}$, then it is generally the case that $\hat{\beta}_{n,M'}(1) \neq \hat{\beta}_{n,M}(1)$ and $\beta_{n,M'}(1) \neq \beta_{n,M'}(1)$, that is, estimates and parameters in submodels are \emph{not} subvectors of their analogues in larger models, except, for example when the columns of $X(M)$ are orthogonal. The reason for this is the collinearity between the covariates in model $M$. The comments above applies for general models $M'\subset M$.
This is why we must write $M$ as a subscript and not in parentheses. (See Section 3.1 of \cite{Berk13} for a related discussion.)

Next define related matrices and vectors as follows:
\begin{equation} \label{eq:MatrixVector}
\begin{split}
\hat{\Sigma}_n := \frac{1}{n}\sum_{i=1}^n X_iX_i^{\top}\in\mathbb{R}^{p\times p},\quad&\mbox{and}\quad \hat{\Gamma}_n := \frac{1}{n}\sum_{i=1}^n X_iY_i\in\mathbb{R}^{p},\\
\Sigma_n := \frac{1}{n}\sum_{i=1}^n \mathbb{E}\left[X_iX_i^{\top}\right]\in\mathbb{R}^{p\times p},\quad&\mbox{and}\quad \Gamma_n := \frac{1}{n}\sum_{i=1}^n \mathbb{E}\left[X_iY_i\right]\in\mathbb{R}^{p}.
\end{split}
\end{equation}
Note that for these quantities there is no need to define separate versions in submodels $M$ because they are just the submatrices $\hat{\Sigma}_n(M)$ and $\Sigma_n(M)$ and subvectors $\hat{\Gamma}_n(M)$ and $\Gamma_n(M)$, respectively. The OLS estimate of the slope vector and its target in the sub-model $M$ satisfy the following normal equations:
\begin{equation}\label{eq:MatrixVector-OLS}
\hat{\Sigma}_n(M) \hat{\beta}_{n,M} = \hat{\Gamma}_n(M), ~~~~
\Sigma_n(M) \beta_{n,M} = \Gamma_n(M).
\end{equation}
\begin{rem}
We do not solve the equations \eqref{eq:MatrixVector-OLS} on purpose because the confidence regions to be constructed below will accommodate exact collinearity by including subspaces of degeneracy.  Minimizers of the objective functions $\hat{R}_n(\theta; M)$ and $R_n(\theta; M)$ defined in \eqref{eq:EmpObj} and \eqref{eq:PopObj} always exist, even if they are not unique. Estimates $\hat{\beta}_{n,M}$ can only be unique when $|M| \le n$ because $\hat{\Sigma}_n(M)$ has rank at most~$\min\{|M|, n\}$. Targets $\beta_{n,M}$, on the other hand, can be unique without a constraint on $n$ because they are based on expectations rather than finite averages, so $\Sigma_n$ and $\Sigma_n(M)$ can be strictly positive definite and $R_n(\theta; M)$ strictly convex with a unique minimizer even when $|M| > n$.
\end{rem}


\subsection{Problem Formulation} \label{sec:problem_formulation}

Under very mild assumptions, $\hat{\beta}_{n,M} - \beta_{n,M}$ converges to zero as $n$ tends to infinity for any fixed, non-random model $M$ (see \cite{Chap2:Kuch18}). This fact justifies calling $\hat{\beta}_{n,M}$ an estimator of $\beta_{n,M}$ or, equivalently, $\beta_{n,M}$ the target of estimation of $\hat{\beta}_{n,M}$. Also, for a fixed $M$, $\hat{\beta}_{n,M}$ has an asymptotic normal distribution, i.e.,
\[
n^{1/2}\left(\hat{\beta}_{n,M} - \beta_{n,M}\right) ~\overset{\mathcal{L}}{\to}~ N\left(0, AV_M\right)~~~~(0 \in \mathbb{R}^{|M|},~AV_M \in \mathbb{R}^{|M|\times|M|} )
\]
for some positive definite matrix $AV_M$ that depends on $M$ and some moments of $(X, Y)$; see the linear representation in \cite{Chap2:Kuch18}. The notation $\overset{\mathcal{L}}{\to}$ denotes convergence in law (or distribution). 
Asymptotic normality lends itself for the construction of $(1 \!-\! \alpha)$-confidence regions $\hat{\mathcal{R}}_{n,M}$ such that
\[
\liminf_{n\to\infty}\,\mathbb{P}\left(\beta_{n,M} \in \hat{\mathcal{R}}_{n,M}\right) \ge 1- \alpha
\]
for any fixed $\alpha\in[0,1]$. We approach statistical inference using confidence regions rather than statistical tests, but this is a technical rather than a conceptual choice because confidence regions and tests are in a duality to each other: a confidence region with coverage at least $1\!-\!\alpha$ is a set of parameter values that could not be rejected at level $\alpha$ if used as point null hypotheses.


The problem of valid post model-selection inference is to construct for given non-random sets of models $\mathcal{M}_p$ a set of confidence regions $\{\hat{\mathcal{R}}_{n,M}\!: M \!\in\! \mathcal{M}_p\}$ such that for any {\bf\em random model} $\hat{M}$ depending (possibly) on the same data satisfying $\mathbb{P}\left(\hat{M} \!\in\! \mathcal{M}_p\right) \!=\! 1$, we have
\begin{equation}\label{eq:PoSIRequire}
\liminf_{n\to\infty}\,\mathbb{P}\left(\beta_{n,\hat{M}} \in \hat{\mathcal{R}}_{n,\hat{M}}\right) \ge 1 - \alpha.
\end{equation}
The guarantee \eqref{eq:PoSIRequire} requires the confidence asymptotically because we strive for a theory that requires few assumptions, whereas finite sample confidence guarantees require strong assumptions.

The notation $\hat{M}$ for random models requires an elaboration of the sources of randomness envisioned here.
With the reproducibility crisis in mind, we cast a wide net for the sources of model randomness by adopting a broad frequentist perspective that includes not only datasets but data analysts as well. Conventional frequentism can be conceived as capturing the random nature of an observed dataset in the actual world by embedding it in a universe of possible worlds with datasets characterized by a joint probability distribution of the observations. We broaden the concept by pairing the random datasets with random data analysts who have varying data analytic preferences and backgrounds. This variability among data analysts may be called ``random researcher degrees of freedom'', a term that alludes to the freedoms we exercise when analyzing in general, and when selecting covariates in a regression in particular. Some of the latter freedoms have been described and classified by \cite{Berk13}, Section~1: (1)~formal selection methods such as stepwise forward or backward selection, lasso-based selection using a criterion to select the penalty parameter, or all-subset search using a criterion such as $C_p$, AIC, BIC, RIC,~etc.; (2)~informal selection steps such as examination of residual plots or influence measures to judge acceptability of models; (3)~post hoc selection such as making substantive trade-offs of predictive viability versus cost of data collection. The waters get further muddied even in the case of formal selection methods (1) when ``informal meta-selection'' is exercised: trying out multiple formal selection methods, comparing them, and favoring some over others based on the results produced on the data at hand. This list of ``researcher degrees of freedom'' in model selection should make it evident that these freedoms are indeed exercised in practice, but in ways that should be called ``subjective'', namely, based on personal background, experience and motivations, as well as historic and institutional contexts. For these reasons it may be infeasible to capture the randomness contributed by data analysts' exercise of their freedoms in terms of stochastic models.

Following \cite{Berk13}, this infeasibility can be bypassed by adding a quantifier ``for all $\hat{M}$'' to the requirement \eqref{eq:PoSIRequire}, thereby capturing all possible ways in which selection may be performed. The added gain is that at a technical level the requirement \eqref{eq:PoSIRequire} permits a reduction to a problem of simultaneous inference.

We must, however, impose certain limits on the freedom of model selection: The set of potential regressors must be pre-specified before examining the data. For example, it is not permissible to initially declare the regressors $X_1,\ldots,X_p$ to be the universe for searching submodels, only to decide after looking at the data that one would also like to search among product interactions $X_j X_k$. The decision to include interactions in data-driven selection would have to be made before looking at the data. Thus data-driven expansion of the universe of regressors for selection is not covered by our framework.

Again following \cite{Berk13}, a curious aspect of the target of estimation has to be noted: $\beta_{n,\hat{M}}$ has become a random quantity with a random dimension $|\hat{M}|$, whereas for a fixed $M$ the target $\beta_{n,{M}}$ is a constant. After data-driven modeling the selected target $\beta_{n,\hat{M}}$ has become random due to data-driven selection $\hat{M}$. This, however, is the only randomness present: among all possible targets $\{ \beta_{n,M} : M \in \mathcal{M}_p \}$, one is randomly selected, namely,~$\beta_{n,\hat{M}}$. The associated estimate $\hat{\beta}_{n,\hat{M}}$ in the random model $\hat{M}$, in addition to its intrinsic variability, also incurs the randomness due to selection. On a technical level, note that the random target $\beta_{n,\hat{M}}$ for the random selection $\hat{M}$ may exist even if the estimate $\hat{\beta}_{n,\hat{M}}$ may not exist due to collinearity. This issue requires some care in Lemmas \ref{lem:UniformConsisL1} and \ref{lem:RateD1nD2n} below.

The inference criterion in \eqref{eq:PoSIRequire} can be decomposed by conditioning on the data-driven selections:
\begin{equation} \label{eq:MarginalConditional}
\mathbb{P}\left(\beta_{n,\hat{M}}\in\hat{\mathcal{R}}_{\hat{M}}\right) ~=~
\sum_{M \in \mathcal{M}_p}
\mathbb{P}\left(\beta_{n,M} \in \hat{\mathcal{R}}_{M} ~\bigg|~ \hat{M} = M\right)
\mathbb{P}\left(\hat{M} = M\right) .
\end{equation}
Plainly, if a guarantee of the form \eqref{eq:PoSIRequire} is available for the marginal probability on the left hand side, no guarantee can be deduced for the conditional probabilities given the random events $\hat{M} = M$ on the right hand side. The decomposition \eqref{eq:MarginalConditional} makes explicit the difference between our current marginal approach and the approach taken by \cite{Lee16}, \cite{Tibs16} and \cite{Tian16}, for example.

We mention briefly that \cite{Rin16} use a notion of ``honest confidence'' that asks for valid inference uniformly over a class of data-generating distributions, that is,
\[
\liminf_{n\to\infty}\,\inf_{\mathbb{P}\in\mathcal{P}_n}\mathbb{P}\left(\beta_{n,\hat{M}}\in\hat{\mathcal{R}}_{n,\hat{M}}\right) \ge 1 - \alpha,
\]
for some class of probability distributions $\mathcal{P}_n$ of the observations. This ``honesty'' holds for our results, too, due to the uniform validity of the multiplier bootstrap proved by \cite{Chern17}, but we will not discuss this further.

\subsection{Alternative Approaches}

There exists an ``obvious'' approach to valid post-selection inference based on sample splitting, as examined by \cite{Rin16}: split the data into two disjoint parts, then use one part for selecting a model $\hat{M}$ and the other part for inference in the selected model $\hat{M}$.  If the two parts of the data are stochastically independent of each other, post-selection inferences will be valid.  For independent observations \cite{Rin16} were able to provide very general and powerful results.  Sample splitting has considerable appeal due to its universal applicability under independence of the two parts: it ``works'' for any type of model selection, formal or informal, as well as for any type of model being fitted. It has some drawbacks, too, an obvious one being the reduced sample sizes of the two parts, which increase the sampling variability of both the model selection stage and the inference stage. Another drawback is that required independence of the two parts, which makes it less obvious how to generalize sample splitting to dependent data. For customers of statistical inferences, it may also be somewhat disconcerting to realize that the splitting procedure incurs a level of artificial randomness and might have produced different results in the hands of another data analyst who would have used another random split. Reliance on random splits brings to our attention a greater concern that relates to the reproducibility crisis in the sciences: sample splitting introduces another ``researcher degree of freedom'', namely, the freedom to choose a particular split after having tried several splits. In practice it would seem extremely unrealistic to assume that data analysts will in fact commit themselves to using just one random split and not be tempted to try several. It could even be argued that using just one split would be irresponsible because it throws away a chance to learn about the stability of model selection and subsequent inferences under multiple splits. Having performed such a stability analysis, however, invalidates the post-selection inferences obtained from the splits because another level of selection arises: that of choosing one of the splits for final reporting. This would not be a problem if stability analysis showed that the same model is being selected in the vast majority of splits, but experience with regression shows that this is not the generic situation: In most regressions, there exist large numbers of submodels with nearly identical performance, making it likely that model selection will be highly variable between sample splits. In summary, while high in intuitive appeal, sample splitting opens up another pandoras box of selection possibilities that may defeat the solution it was meant to provide.

A different type of post-selection guarantees are available from the approach of \cite{Lee16}, \cite{Tibs16} and \cite{Tian16} when model selection is of a pre-specified form such as lasso selection or stepwise forward selection. The inference guarantees they provide are conditional on the selected model. Their approach is ingeniously tailored to specific formal selection methods and takes advantage of their properties. It is, however, a model-trusting approach that relies much on the correctness of the assumed model as being finite-sample correct under a Gaussian linear model with fixed covariates. For this reason and because so much conditioning is performed, it is unlikely that this approach enjoys much robustness to misspecification (see, for example, Section A.20 of \cite{MR3798003}). 
By comparison, we strive here for model robustness by limiting ourselves to asymptotically correct coverage that is marginal rather than conditional, and by allowing covariates to be treated as random rather than fixed.

A larger point to be reiterated here is that tailoring post-selection inference to a specific formal selection method such as the lasso does not address the issue that data analysts may not limit themselves to just one formal selection method and nothing else. It may be more realistic to assume, as we do here, that they exercise broader liberties that include trying out multiple formal selection methods as well as informal model selection of various kinds. Providing and recommending valid post-selection inference that casts a wider net on selection methods may have a better chance of making an at least partial contribution to solving the reproducibility crisis in the sciences.


\section{Equivalence of Post-selection and Simultaneous Inference}\label{sec:Simultaneous}

The first step towards achieving the goal of constructing a set of confidence regions $\{\hat{\mathcal{R}}_{n,M}:\, M\in\mathcal{M}_p\}$ satisfying \eqref{eq:PoSIRequire} is to convert the post-selection inference problem into a simultaneous inference problem. This conversion is provided by Theorem \ref{thm:Uniform}, which parallels \cite{Berk13} but offers the generality needed here. The theorem is proved for finite samples, but a version using ``$\liminf$'' follows readily.

\begin{thm}\label{thm:Uniform}
For any set of confidence regions $\{\hat{\mathcal{R}}_{n,M}:\,M\in\mathcal{M}_p\}$ and $\alpha\in[0,1]$, the following two statements are equivalent:
\begin{itemize}
\item[$(1)$] The post-selection inference problem is solved, that is,
\[
\mathbb{P}\left(\beta_{n,\hat{M}}\in\hat{\mathcal{R}}_{n,\hat{M}}\right) \ge 1 - \alpha,
\]
for all data-dependent model selections satisfying $\mathbb{P}(\hat{M}\in\mathcal{M}_p) = 1$.
\item[$(2)$] The simultaneous inference problem over $M \!\in\! \mathcal{M}_p$ is solved, that is,
\[
\mathbb{P}\left(\bigcap_{M\in\mathcal{M}_p}\left\{\beta_{n,M}\in\hat{\mathcal{R}}_{n,M}\right\}\right) \ge 1 - \alpha.
\]
\end{itemize}
\end{thm}
\begin{proof}
Define for any fixed $M\!\in\!\mathcal{M}_p$ the coverage event $\mathcal{A}_M \!=\! \{\beta_{n,M}\in\hat{\mathcal{R}}_{n,M}\}$, and similarly $\mathcal{A}_{\hat{M}} \!=\! \{\beta_{n,\hat{M}} \in \hat{\mathcal{R}}_{n,\hat{M}}\}$. Note that $\mathcal{A}_{\hat{M}}$ is the event in $(1)$ and $\bigcap_{M\in\mathcal{M}_p} \mathcal{A}_M$ the event in (2).

$(2)\Rightarrow(1)$: It is sufficient to show that for any random selection procedure $\hat{M}$ we have
\[
\bigcap_{M\in\mathcal{M}_p} \mathcal{A}_M ~\subseteq~ \mathcal{A}_{\hat{M}}.
\]
Because $\hat{M}$ takes on values in $\mathcal{M}_p$ only, $\bigcup_{M' \in \mathcal{M}_p} \{ \hat{M} = M' \}$ is the whole sample space.  Hence
\begin{align*}
\mathcal{A}_{\hat{M}} &= \bigcup_{M^{\prime}\in\mathcal{M}_p} \{\hat{M} = M^{\prime}\} \cap \mathcal{A}_{M^{\prime}} \\
&\supseteq \bigcup_{M^{\prime}\in\mathcal{M}_p} \{\hat{M} = M^{\prime}\} \cap \bigcap_{M\in\mathcal{M}_p}\mathcal{A}_M \\
&= \bigcap_{M\in\mathcal{M}_p}\mathcal{A}_M .
\end{align*}

$(1)\Rightarrow(2)$: To prove this implication, it is sufficient to construct a data-driven (hence random) selection procedure $\hat{M}$ that satisfies
\begin{equation} \label{eq:worst-Mhat}
\mathcal{A}_{\hat{M}} ~= \bigcap_{M\in\mathcal{M}_p} \mathcal{A}_M .
\end{equation}
This is achieved by letting $\hat{M}$ be any selection procedure that satisfies
\[
\hat{M} ~\in~ \argmin_{M\in\mathcal{M}_p}\, \mathbbm{1}\{\mathcal{A}_M\} ,
\]
where $\mathbbm{1}\{A\}$ denotes the indicator of event $A$. It follows that
\[
\mathbbm{1}\{\mathcal{A}_{\hat{M}}\} ~=~ \min_{M\in\mathcal{M}_p}\,\mathbbm{1}\{\mathcal{A}_M\} ,
\]
which is equivalent to \eqref{eq:worst-Mhat}. This completes the proof of $(1)\Rightarrow(2)$.
\end{proof}

\begin{rem}\label{rem:theorem1-proof}
The proof makes no use of the regression context at all; it is merely about indexed sets/events $\mathcal{A}_M$ and random selections $\hat{M}$ of the indexes~$M$. The second part of the proof constructs an adversarial random selection procedure $\hat{M}$ that requires simultaneous coverage over all~$M$.
\end{rem}
\begin{rem}\label{rem:theorem1-meaning}
The theorem establishes the equivalence of family-wise simultaneous coverage and post-selection coverage allowing for arbitrary random selection. The argument, because it makes no use of the regression context, applies to any type of regression.
\end{rem}
\begin{rem}\label{rem:theorem1-Berk13}
Lemma~4.1 in \cite{Berk13} (``Significant triviality bound'') corresponding to Theorem~\ref{thm:Uniform} is much more intuitive because it is based on maxima over pivotal $t$-statistics rather than confidence regions. The gain in intuition, however, is purchased at a price: an injection of mathematically irrelevant detail. The bare-bones nature of the underlying structure is revealed by the above proof which does not even involve probability but set theory only.
\end{rem}
\begin{rem}\;(Inherent High-dimensionality)\label{rem:highdim}
Returning to regression, note that in view of Theorem~\ref{thm:Uniform}, valid post-selection inference is inherently a high-dimensional problem in the sense that the number of parameters subject to estimation and inference is large, indeed, often larger than the sample size. For illustration, consider a common regression setting where the number of covariates is $p = 10$ and the sample of size $n = 500$. Estimation and testing of the slopes in the full model seems unproblematic because there are 50 observations per parameter. Now, for the post-selection inference problem with all non-empty sub-models, there are $2^{p} - 1 = 1023$ vector parameters of varying dimensions, adding up to a total of $p2^{p-1} = 5120$ parameters in the various submodels, exceeding the sample size $n=500$ by a factor of ten and thus constituting an inference problem in the high-dimensional category.
\end{rem}
Theorem \ref{thm:Uniform} shows that in order to achieve universally valid post-selection inference, that is, inference that satisfies \eqref{eq:PoSIRequire} {\bf\em for all} data-driven selection procedures $\hat{M}$, it is necessary and sufficient to construct a set of confidence regions $\hat{\mathcal{R}}_{n,M}$ such that
\begin{equation}\label{eq:PoSIRequire2}
\liminf_{n\to\infty}\mathbb{P}\left(\bigcap_{M\in\mathcal{M}_p}\left\{\beta_{n,M} \in \hat{\mathcal{R}}_{n,M}\right\}\right) \ge 1 - \alpha.
\end{equation}
All of our solutions to the post-selection inference problem in this article are constructed to satisfy~\eqref{eq:PoSIRequire2}.


\section{An Approach to Post-Selection Inference}\label{sec:FirstStrat}

\subsection{Valid Confidence Regions} \label{sec:ValidConf}

Equipped with the required notation, we proceed to construct confidence regions $\hat{\mathcal{R}}_{n,M}$ for linear regression. From Equations~\eqref{eq:EmpObj} and~\eqref{eq:PopObj}, we see that the least squares estimator and target given in \eqref{eq:LeastSquares} can be written as
\begin{equation}\label{eq:ModifiedObjective}
\begin{split}
\hat{\beta}_{n,M} &= \argmin_{\theta\in\mathbb{R}^{|M|}}\, \left\{\theta^{\top}\hat{\Sigma}_n(M)\theta - 2\theta^{\top}\hat{\Gamma}_n(M)\right\},\quad \mbox{and}\\
\beta_{n,M} &= \argmin_{\theta\in\mathbb{R}^{|M|}}\, \left\{\theta^{\top}\Sigma_n(M)\theta - 2\theta^{\top}\Gamma_n(M)\right\}.
\end{split}
\end{equation}
The differences between the two objective functions in \eqref{eq:ModifiedObjective} can be controlled in terms of two error norms below related to the $\Sigma$ matrices and the $\Gamma$ vectors defined in \eqref{eq:MatrixVector}. Define therefore the estimation errors of $\hat{\Sigma}_n$ and $\hat{\Gamma}_n$ as follows:
\begin{equation}\label{eq:D1nD2n}
\begin{split}
\mathcal{D}_{n}^{\Sigma} &~:=~ \norm{\hat{\Sigma}_n - \Sigma_n}_{\infty} = \max_{M\in\mathcal{M}_p(2)}\norm{\hat{\Sigma}_n(M) - \Sigma_n(M)}_{\infty},
\\
\mathcal{D}_{n}^{\Gamma} &~:=~ \norm{\hat{\Gamma}_n - \Gamma_n}_{\infty} = \max_{M\in\mathcal{M}_p(1)}\norm{\hat{\Gamma}_n(M) - \Gamma_n(M)}_{\infty}.
\end{split}
\end{equation}
The equalities on the right are useful trivialities given here for later use: $\mathcal{M}_p(2)$ and $\mathcal{M}_p(1)$ are the sets of all models of sizes bounded by 2 and 1, respectively, where size 1 is sufficient for ``$\max$'' to reach all elements of the $\Gamma$ vectors, but size 2 is needed for ``$\max$'' to reach all off-diagonal elements of the $\Sigma$ matrices as well. Importantly, neither $\mathcal{D}_{n}^{\Sigma}$ nor $\mathcal{D}_{n}^{\Gamma}$ is a function of submodels~$M$.

The quantities $\mathcal{D}_{n}^{\Sigma}$ and $\mathcal{D}_{n}^{\Gamma}$ are statistics whose quantiles will play an essential role in the construction of the confidence regions to be defined next.  In each submodel $M \in \mathcal{M}_p(p)$, we will construct for the parameter vector $\beta_{n,M}$ two confidence regions: The first satisfies {\em finite sample} guarantees at the cost of lesser transparency, whereas the second satisfies {\em asymptotic} guarantees with the benefit of greater simplicity. The motivations for the particular forms of these regions will become clear in the course of the elementary proofs of the theorems to follow.  With these preliminary remarks in mind, we define
\begin{align}
\hat{\mathcal{R}}_{n,M} &:= \left\{\theta\in\mathbb{R}^{|M|}:\, \norm{\hat{\Sigma}_n(M)\left\{\hat{\beta}_{n,M} - \theta\right\}}_{\infty} \le C_{n}^{\Gamma}(\alpha) + C_{n}^{\Sigma}(\alpha)\norm{\theta}_1\right\},
\label{eq:FirstFinite}\\
\hat{\mathcal{R}}_{n,M}^{\dagger} &:= \left\{\theta\in\mathbb{R}^{|M|}:\, \norm{\hat{\Sigma}_n(M)\left\{\hat{\beta}_{n,M} - \theta\right\}}_{\infty} \le C_{n}^{\Gamma}(\alpha) + C_{n}^{\Sigma}(\alpha)\norm{\hat{\beta}_{n,M}}_1\right\},
\label{eq:FirstAsym}
\end{align}
where $C_{n}^{\Gamma}(\alpha)$ and $C_{n}^{\Sigma}(\alpha)$ are bivariate joint quantiles of $\mathcal{D}_{n}^{\Gamma}$ and $\mathcal{D}_{n}^{\Sigma}$ in \eqref{eq:D1nD2n}, that is,
\begin{equation}\label{eq:ConfidenceR}
\mathbb{P}\left(\mathcal{D}_{n}^{\Gamma} \le C_{n}^{\Gamma}(\alpha) ~~\textrm{and}~~ \mathcal{D}_{n}^{\Sigma} \le C_{n}^{\Sigma}(\alpha)\right) ~\ge~ 1 - \alpha.
\end{equation}
\begin{rem}{ (Restriction of Models for Selection)}
The confidence regions defined in \eqref{eq:FirstFinite} and \eqref{eq:FirstAsym} do not take advantage of restricted model universes such as ``sparse model selection'' where $\hat{M} \in \mathcal{M}_p(k)$ searches only models of sizes up to $k~(< p)$. It might, however, be of practical interest to consider the post-selection inference problem when the set of models used in selection is indeed a strict subset of the set $\mathcal{M}_p(p)$ of all models. This can be accommodated with an obvious tweak whereby
\[
\mathcal{D}_{n}^{\Gamma}(\mathcal{M}_p) := \sup_{M\in\mathcal{M}_p}\norm{\hat{\Gamma}_n(M) - \Gamma_n(M)}_{\infty}\quad\mbox{and}\quad \mathcal{D}_{n}^{\Sigma}(\mathcal{M}_p) := \sup_{M\in\mathcal{M}_p}\norm{\hat{\Sigma}_n(M) - \Sigma_n(M)}_{\infty}
\]
become functions of the restricted model universe $\mathcal{M}_p ~(\subsetneq\mathcal{M}_p(p))$. Note, however, that according to~\eqref{eq:D1nD2n} we have $\mathcal{D}_{n}^{\Gamma}(\mathcal{M}_p)  = \mathcal{D}_{n}^{\Gamma}$ as long as the model universe $\mathcal{M}_p$ includes all models of size one, and $\mathcal{D}_{n}^{\Sigma}(\mathcal{M}_p) = \mathcal{D}_{n}^{\Sigma}$ as long as $\mathcal{M}_p$ includes all models of size two. This is the case, for example, when ``sparse model selection'' is used, meaning $\mathcal{M}_p = \mathcal{M}_p(k)$ for $k < p$. Thus confidence regions of the form \eqref{eq:FirstFinite} do not gain from ``sparse model selection.'' This is so because the regions depend effectively only on marginal and bivariate properties of the observations~$(X_i,Y_i)$ and their distributions through $\Gamma_n$, $\hat{\Gamma}_n$, $\Sigma_n$ and $\hat{\Sigma}_n$.
\end{rem}
\medskip
\noindent Further observations on $(\mathcal{D}_n^{\Gamma}, \mathcal{D}_n^{\Sigma})$ and $(C_n^{\Gamma}(\alpha), C_n^{\Sigma}(\alpha))$:
\begin{itemize} \itemsep 0em
\item Bivariate quantiles are not unique: one may marginally increase one and decrease the other suitably, maintaining the bivariate coverage probability $1-\alpha$. Allowed is any choice of $C_{n}^{\Gamma}(\alpha)$ and $C_{n}^{\Sigma}(\alpha)$ that satisfies~\eqref{eq:ConfidenceR}. 
\item These quantiles are not known and must be estimated from the data. A bootstrap procedure to estimate them is described in Section~\ref{sec:Comp}. 
\item The estimation errors $\mathcal{D}_{n}^{\Gamma}$ and $\mathcal{D}_{n}^{\Sigma}$, being based on averages of quantities of dimensions $p \times 1$ and $p \times p$, respectively, converge by the law of large numbers to zero as $n\to\infty$ under mild conditions (see Lemma~\ref{lem:RateD1nD2n}). Therefore, $C_{n}^{\Gamma}(\alpha)$ and $C_{n}^{\Sigma}(\alpha)$ converge to zero as~$n\to\infty$.
\end{itemize}


\subsection{Validity of the Confidence Regions $\hat{\mathcal{R}}_{n,M}$}

We proceed to proving validity of the simultaneous inference guarantee \eqref{eq:PoSIRequire2}. This will be done in Theorem \ref{thm:Appr1.2} for the confidence regions $\hat{\mathcal{R}}_{n,M}$ where $M \in \mathcal{M}_p(p)$, and in Theorem \ref{thm:PoSIX} for the confidence regions $\hat{\mathcal{R}}^{\dagger}_{n,M}$ where $M \in \mathcal{M}_p(k)$ for some $k \le p$.
\begin{thm}\label{thm:Appr1.2}
The set of confidence regions $\{\hat{\mathcal{R}}_{n,M}:\,M\in\mathcal{M}_p(p)\}$ defined in \eqref{eq:FirstFinite} satisfies
\begin{equation}\label{eq:SimultaneousG1.2}
\mathbb{P}\left(\bigcap_{M\in\mathcal{M}_p(p)}\left\{\beta_{n,M} \in \hat{\mathcal{R}}_{n,M}\right\}\right) \ge 1 - \alpha,
\end{equation}
Furthermore, for any random model $\hat{M}$ with $\mathbb{P}(\hat{M}\in\mathcal{M}_p(p)) = 1$, we have
\[
\mathbb{P}\left(\beta_{n,\hat{M}}\in\hat{\mathcal{R}}_{n,\hat{M}}\right) \ge 1 - \alpha.
\]
\end{thm}

As mentioned earlier, this theorem is non-asymptotic as it provides guarantees for finite samples. It is, however, not directly actionable because, as mentioned earlier also, the bivariate quantiles used in the construction of the confidence regions need to be estimated. Hence actionable versions of these regions end up having only asymptotic guarantees as well.

\begin{proof}
The proof is surprisingly elementary and involves simple manipulation of the estimating equations. We start by subtracting the normal equations of the target from those of the estimates, see \eqref{eq:MatrixVector-OLS}. This holds for all $M \in \mathcal{M}_p(p)$:
\begin{align}
\hat{\Sigma}_n(M)\hat{\beta}_{n,M} - \Sigma_n(M)\beta_{n,M}
&~=~ \hat{\Gamma}_n(M) - \Gamma_n(M) .
\end{align}
Telescope the left side by subtracting and adding $\hat{\Sigma}_n(M) \beta_{n,M}$:
\begin{align}
\hat{\Sigma}_n(M) \left( \hat{\beta}_{n,M} - \beta_{n,M} \right)
+
\left( \hat{\Sigma}_n(M)- \Sigma_n(M) \right) \beta_{n,M}
&~=~ \hat{\Gamma}_n(M) - \Gamma_n(M) ,
\end{align}
Move the second summand on the left to the right side of the equality, take the sup norm and apply the triangle inequality on the right side:
\begin{align}
\left\| \hat{\Sigma}_n(M) \left( \hat{\beta}_{n,M} - \beta_{n,M} \right) \right\|_\infty
&~\le~
\left\| \hat{\Gamma}_n(M) - \Gamma_n(M) \right\|_\infty +
\left\| \left( \hat{\Sigma}_n(M)- \Sigma_n(M) \right) \beta_{n,M} \right\|_\infty ,
\end{align}
Applying the second inequality in \eqref{eq:MatrixVectorIneq} to the last term it follows that
\begin{equation}\label{eq:Conserv2}
\norm{\hat{\Sigma}_n(M)\left\{\hat{\beta}_{n,M}-\beta_{n,M}\right\}}_{\infty} \le \norm{\hat{\Gamma}_n(M) - \Gamma_n(M)}_{\infty} + \norm{\hat{\Sigma}_n(M) - \Sigma_n(M)}_{\infty}\norm{\beta_{n,M}}_{1}.
\end{equation}
Because $\hat{\Gamma}_n(M) - \Gamma_n(M)$ and $\hat{\Sigma}_n(M) - \Sigma_n(M)$ are a subvector and a submatrix of $\hat{\Gamma}_n - \Gamma_n$ and $\hat{\Sigma}_n - \Sigma_n$, respectively, this inequality implies
\begin{equation}\label{eq:MainConserv}
\norm{\hat{\Sigma}_n(M)\left\{\beta_{n,M} - \hat{\beta}_{n,M}\right\}}_{\infty} \le \norm{\hat{\Gamma}_n - \Gamma_n}_{\infty} + \norm{\hat{\Sigma}_n - \Sigma_n}_{\infty}\norm{\beta_{n,M}}_{1}.
\end{equation}
This inequality is deterministic and holds for any sample. It also holds for all $M \in \mathcal{M}_p(p)$. These facts allow us to take the intersection of the events \eqref{eq:MainConserv} over all submodels $M$ and transform it into a ``probability one'' statement.  Using $\mathcal{D}_{n}^{\Gamma}$ and $\mathcal{D}_{n}^{\Sigma}$ defined in \eqref{eq:D1nD2n}, we have
\begin{equation}\label{eq:MainProb1}
\mathbb{P}\left(\bigcap_{M\in\mathcal{M}_p(p)}\left\{\norm{\Sigma_n(M)\left\{\beta_{n,M} - \hat{\beta}_{n,M}\right\}}_{\infty}
~\le~
\mathcal{D}_{n}^{\Gamma} + \mathcal{D}_{n}^{\Sigma}\norm{\beta_{n,M}}_{1}\right\}\right)
~=~ 1.
\end{equation}
From the definitions of $C_{n}^{\Gamma}(\alpha)$ and $C_{n}^{\Sigma}(\alpha)$ in \eqref{eq:ConfidenceR} follows the required result \eqref{eq:SimultaneousG1.2}. The second result of post-selection guarantees for random models follows by an application of Theorem~\ref{thm:Uniform}.
\end{proof}
\begin{rem}{ (Reach of the Validity Guarantee)}
It is interesting to note that the guarantee \eqref{eq:SimultaneousG1.2} in Theorem \ref{thm:Appr1.2} is valid for every sample size $n$ and any number of covariates $p$. In particular, $p \gg n$ and $p = \infty$ are covered without difficulty even though $\hat{\Sigma}_n(M)$ is necessarily singular for $|M| > n$. For this to make sense recall that for singular $\hat{\Sigma}_n(M)$ the confidence region $\hat{\mathcal{R}}_{n,M}$ simply contains a non-trivial affine subspace of~$\reals^p$.
\end{rem}
\begin{rem} { (Estimation of Bivariate Quantiles)}
The finite sample guarantee \eqref{eq:SimultaneousG1.2} requires the bivariate quantiles $C_{n}^{\Gamma}(\alpha)$ and $C_{n}^{\Sigma}(\alpha)$ of $\mathcal{D}_{n}^{\Gamma}$ and $\mathcal{D}_{n}^{\Sigma}$, respectively, to satisfy \eqref{eq:ConfidenceR} for all $p, n \ge 1$. In general, these bivariate quantiles can only be estimated consistently in the asymptotic sense as explained in Section~\ref{sec:Comp}.
\end{rem}
\begin{rem}{ (Independence of Observations)}\label{rem:Indep}
For simplicity in the discussion above, we used the assumption of independence of random vectors $(X_i, Y_i), 1\le i\le n$. Theorem \ref{thm:Appr1.2} holds without this assumption because no use of this assumption was made in its proof. However, validity of the post-selection guarantee holds as long as $C_{n}^{\Gamma}(\alpha)$ and $C_{n}^{\Sigma}(\alpha)$ are valid quantiles in the sense of~\eqref{eq:ConfidenceR}. 
\end{rem}


\subsection[Asymptotic Validity of dagger]{Asymptotic Validity of the Confidence Regions $\hat{\mathcal{R}}_{n,M}^{\dagger}$}

The confidence region $\hat{\mathcal{R}}_{n,M}$ is difficult to analyze in terms of its shape and its Lebesgue measure. (However, with a different parametrization of $\hat{\mathcal{R}}_{n,M}$, \cite{Belloni17} prove that this confidence region is a convex polyhedron; see Equation (42) of the supplement of \cite{Belloni17}.) Because of these difficulties we also prove asymptotic validity of more intuitive confidence regions of the form $\hat{\mathcal{R}}_{n,M}^{\dagger}$ defined in \eqref{eq:FirstAsym}. Because these regions depend on estimates $\hat{\beta}_{n,M}$ whose variability explodes under increasing collinearity, we need to control the minimum eigenvalue of the matrix $\Sigma_n(M)$ for models up to size $k$ to preclude too much collinearity in the limit:
\begin{equation}\label{eq:Lambdak}
\Lambda_n(k) := \min_{M\in\mathcal{M}_p(k)}\lambda_{\min}(\Sigma_n(M)).
\end{equation}
We then make use of the following assumption:

\begin{description}
\item[\namedlabel{eq:UniformConsis}{(A1)(k)}] 
The estimation error $\mathcal{D}_{n}^{\Sigma}$ satisfies
\[
k\mathcal{D}_{n}^{\Sigma} = o_p\left(\Lambda_n(k)\right)\quad\mbox{as}\quad n\to\infty.
\]
\end{description}
This assumption is used for uniform consistency of the least squares estimator in $\norm{\cdot}_1$-norm as in Lemma \ref{lem:UniformConsisL1}. The rate of convergence of $\mathcal{D}_{n}^{\Sigma}$ to zero implies a rate constraint on $k$. Here, as before, $k = k_n$ is allowed to be a sequence depending on $n$. As can be expected, the dependence structure between the random vectors $(X_i, Y_i), 1\le i\le n$ and their moments determine the rate at which $\mathcal{D}_{n}^{\Sigma}$ converges to zero. See Lemma \ref{lem:RateD1nD2n} for more details. The theorem is stated with this high level assumption so that it is more widely applicable in particular to various structural dependencies on observations. Note that assumption \ref{eq:UniformConsis} allows for the minimum eigenvalue of $\Sigma_n$ to converge to zero or even be zero as $n\to\infty$ if $p = p_n$ changes with $n$.

Before proceeding to the proof that $\hat{\mathcal{R}}^{\dagger}_{n,M}$ are asymptotically valid post-selection confidence regions, we prove uniform-in-model consistency of $\hat{\beta}_{n,M}$ to $\beta_{n,M}$. See Appendix \ref{app:ProofLemmaL1} for a detailed proof. Also, see \cite{Uniform:Kuch18} for more results of this flavor.
\begin{lem}\label{lem:UniformConsisL1}
For all $k\ge 1$ satisfying $k\mathcal{D}_{n}^{\Sigma} \le \Lambda_n(k)$ and for all $M\in\mathcal{M}_p(k)$,
\begin{equation}\label{eq:MarginalL1}
\norm{\hat{\beta}_{n,M} - \beta_{n,M}}_1 \le \frac{|M|\left(\mathcal{D}_{n}^{\Gamma} + \mathcal{D}_{n}^{\Sigma}\norm{\beta_{n,M}}_1\right)}{\Lambda_n(k) - k\mathcal{D}_{n}^{\Sigma}}.
\end{equation}
\end{lem}
The following theorem proves the validity of the simultaneous inference guarantee for~$\hat{\mathcal{R}}^{\dagger}_{n,M}$.
\begin{thm}\label{thm:PoSIX}
For every $1\le k\le p$ that satisfies \ref{eq:UniformConsis}, the confidence regions $\hat{\mathcal{R}}_{n,M}^{\dagger}$ defined in \eqref{eq:FirstAsym} satisfy
\[
\liminf_{n\to\infty}\,\mathbb{P}\left(\bigcap_{M\in\mathcal{M}_p(k)}\left\{\beta_{n,M}\in\hat{\mathcal{R}}_{n,M}^{\dagger}\right\}\right) \ge 1 - \alpha.
\]
\end{thm}
\begin{proof}
The starting point of this proof is Equation \eqref{eq:MainProb1}. Under assumption \ref{eq:UniformConsis}, Lemma \ref{lem:UniformConsisL1} (inequality \eqref{eq:MarginalL1}) implies that for all $M\in\mathcal{M}_p(k)$,
\begin{align*}
\left|\frac{\mathcal{D}_{n}^{\Gamma} + \mathcal{D}_{n}^{\Sigma}\norm{\hat{\beta}_{n,M}}_1}{\mathcal{D}_{n}^{\Gamma} + \mathcal{D}_{n}^{\Sigma}\norm{\beta_{n,M}}_1} - 1\right| &\le \frac{\mathcal{D}_{n}^{\Sigma}\norm{\hat{\beta}_{n,M} - \beta_{n,M}}_1}{\mathcal{D}_{n}^{\Gamma} + \mathcal{D}_{n}^{\Sigma}\norm{\beta_{n,M}}_1}\\
&\le \frac{\mathcal{D}_{n}^{\Sigma}}{\mathcal{D}_{n}^{\Gamma} + \mathcal{D}_{n}^{\Sigma}\norm{\beta_{n,M}}_1}\cdot \frac{|M|\left\{\mathcal{D}_{n}^{\Gamma} + \mathcal{D}_{n}^{\Sigma}\norm{\beta_{n,M}}_1\right\}}{\Lambda_n(k) - |M|\mathcal{D}_{n}^{\Sigma}}\\
&\le \frac{k\mathcal{D}_{n}^{\Sigma}}{\Lambda_n(k) - k\mathcal{D}_{n}^{\Sigma}}.
\end{align*}
Therefore, for $1\le k\le p$ satisfying assumption \ref{eq:UniformConsis},
\[
\sup_{M\in\mathcal{M}_p(k)}\left|\frac{\mathcal{D}_{n}^{\Gamma} + \mathcal{D}_{n}^{\Sigma}\norm{\hat{\beta}_{n,M}}_1}{\mathcal{D}_{n}^{\Gamma} + \mathcal{D}_{n}^{\Sigma}\norm{\beta_{n,M}}_1} - 1\right| \le \frac{k\mathcal{D}_{n}^{\Sigma}/\Lambda_n(k)}{1 - \left(k\mathcal{D}_{n}^{\Sigma}/\Lambda_n(k)\right)} = o_p(1).
\]
Hence,
\[
\liminf_{n\to\infty}\,\mathbb{P}\left(\bigcap_{M\in\mathcal{M}_p(k)}\left\{\norm{\Sigma_n(M)\left\{\beta_{n,M} - \hat{\beta}_{n,M}\right\}}_{\infty} \le \mathcal{D}_{n}^{\Gamma} + \mathcal{D}_{n}^{\Sigma}\norm{\hat{\beta}_{M}}_{1}\right\}\right) = 1.
\]
The definition of $(C_{n}^{\Gamma}(\alpha), C_{n}^{\Sigma}(\alpha))$ in \eqref{eq:ConfidenceR} proves the required result.
\end{proof}


\subsection[Further Remarks]{Further Remarks on the Confidence Regions $\hat{\mathcal{R}}_{n,M}$ and $\hat{\mathcal{R}}^{\dagger}_{n,M}$}

\begin{rem}{ (Centering and Scaling)}\label{rem:Invariance}
The confidence regions $\hat{\mathcal{R}}_{n,M}$ and $\hat{\mathcal{R}}^{\dagger}_{n,M}$ are not equivariant with respect to linear transformation of covariates or the response. Equivariance is an important feature for practical interpretation. A simple way to obtain equivariance with respect to diagonal linear transformations of the random vectors would be to use linear regression with covariates centered and scaled to have sample mean zero and sample variance 1. Since the validity of confidence regions does not require independence, as mentioned in Remark \ref{rem:Indep}, this centering and scaling based on the data will not affect the post-selection guarantee as long as marginal means and variances are estimated consistently. This might also have an effect on the volume of the confidence regions not in terms of rate but in terms of constants since the intercept is not longer needed in $\norm{\beta_{n,M}}_1$. See Section \ref{sec:ProsAndCons} for more details.
\end{rem}
\begin{rem}{ (Shape of $\hat{\mathcal{R}}_{n,M}^{\dagger}$)}
The confidence region $\hat{\mathcal{R}}_{n,M}^{\dagger}$ is a polyhedron, because it can be described by $2|M|$ linear inequalities (with random coefficients). More specifically, it is a parallelepiped because the inequalities come in pairs of parallel constraints. The Lebesgue measure of this confidence region is much easier to study than that of the region $\hat{\mathcal{R}}_{n,M}$ (see Proposition~\ref{lem:LebesgueMeasure} below).
\end{rem}
\begin{rem}{ (Comparison of $\hat{\mathcal{R}}_{n,M}$ and $\hat{\mathcal{R}}_{n,M}^{\dagger}$ in Testing)} As mentioned before, the shape of the confidence region $\hat{R}_{n,M}$ is not easily described. There are, however, scenarios where the advantages of $\hat{\mathcal{R}}_{n,M}$ over $\hat{\mathcal{R}}_{n,M}^{\dagger}$ can be clearly understood. Consider the problem of significance testing, that is, $H_{0,M}:\,\beta_{n,M} = 0$. The level $\alpha$ test based on the confidence region $\hat{\mathcal{R}}_{n,M}$ rejects $H_{0,M}$ if
\begin{equation}\label{eq:RejectionRegionFinite}
\norm{\hat{\Sigma}_n(M)\hat{\beta}_{n,M}}_{\infty} \ge C_n^{\Gamma}(\alpha).
\end{equation}
By comparison, the level $\alpha$ test based on the confidence region $\hat{\mathcal{R}}_{n,M}^{\dagger}$ rejects $H_{0,M}$ if
\begin{equation}\label{eq:RejectionRegionAsym}
\norm{\hat{\Sigma}_{n}(M)\hat{\beta}_{n,M}}_{\infty} \ge C_n^{\Gamma}(\alpha) + C_n^{\Sigma}\norm{\hat{\beta}_{n,M}}_1.
\end{equation}
Thus $\hat{\mathcal{R}}_{n,M}$ results in more rejections and hence greater power than $\hat{\mathcal{R}}_{n,M}^{\dagger}$ at the same level~$\alpha$. A similar argument holds even if the null hypothesis is changed to $H_0:\,\beta_{n,M} = \theta_0\in\mathbb{R}^{|M|}$ for some sparse~$\theta_0$. 
\end{rem}
\subsection[Rate Bounds and Lebesgue Measure]{Rate Bounds on $\mathcal{D}_n^{\Gamma}$, $\mathcal{D}_n^{\Sigma}$ and Lebesgue Measure of the Regions}
Before proceeding further with the study of the confidence regions, it might be useful to understand the rates at with $\mathcal{D}_{n}^{\Gamma}$ and $\mathcal{D}_{n}^{\Sigma}$ converge to zero under some assumptions on the initial random vectors $(X_i, Y_i), 1\le i\le n$. As mentioned in Remark \ref{rem:Indep}, the validity of post-selection coverage guarantee does not require independence of random vectors and so, a rate result under ``functional dependence'' is presented in Appendix \ref{app:Dependent}. Set $Z_i = (X_i^{\top}, Y_i)^{\top}$ for $1\le i\le n$ and define
\begin{equation}\label{eq:OmegaDef}
\hat{\Omega}_n := \frac{1}{n}\sum_{i=1}^n Z_iZ_i^{\top},\quad\mbox{and}\quad \Omega_n := \frac{1}{n}\sum_{i=1}^n \mathbb{E}\left[Z_iZ_i^{\top}\right]\in\mathbb{R}^{(p+1)\times(p+1)}.
\end{equation}
Observe that $$\max\{\mathcal{D}_{n}^{\Gamma}, \mathcal{D}_{n}^{\Sigma}\} \le \norm{\hat{\Omega}_n - \Omega_n}_{\infty}.$$
The following lemma from \cite{Uniform:Kuch18} proves a finite sample bound for the expected value of the maximum absolute value of $\hat{\Omega}_n - \Omega_n$. For this result, set for $\gamma > 0$ and any random variable $W$,
\[
\norm{W}_{\psi_{\gamma}} := \inf\left\{C > 0:\,\mathbb{E}\left[\psi_{\gamma}\left(\frac{|W|}{C}\right)\right] \le 1\right\},
\]
where $\psi_{\gamma}(x) = \exp(x^{\gamma}) - 1$ for $x \ge 0$. For $0 < \gamma < 1$, $\norm{\cdot}_{\psi_{\gamma}}$ is not a norm but is a quasi-norm. A random variable $W$ satisfying $\norm{W}_{\psi_{\gamma}} < \infty$ is called a sub-Weibull random variable of order $\gamma$. The special cases $\gamma = 1$ and $\gamma = 2$ correspond to the well-known classes of sub-exponential and sub-Gaussian random variables.
\begin{lem}\label{lem:RateD1nD2n}
Fix $n, p\ge 2$. Suppose the random vectors $Z_i, 1\le i\le n$ are independent and satisfy for some $0 < \gamma \le 2$
\begin{equation}\label{eq:MarginalPhi}
\max_{1\le i\le n}\max_{1\le j\le p+1}\norm{Z_i(j)}_{\psi_{\gamma}} \le K_{n,p},
\end{equation}
for some positive constant $K_{n,p}.$ Then
\begin{equation}\label{eq:ExpMaxBd}
\mathbb{E}\left[\sqrt{n}\norm{\hat{\Omega}_n - \Omega_n}_{\infty}\right] \le C_{\gamma}\left\{A_{n,p}\sqrt{\log p} + {K_{n,p}^2}(\log p\log n)^{2/\gamma}n^{-1/2}\right\},
\end{equation}
and for all $\alpha\in(0,1]$,
\[
\max\{C_n^{\Gamma}(\alpha), C_n^{\Sigma}(\alpha)\} \le 7A_{n,p}\sqrt{\frac{\log\left(\frac{3}{\alpha}\right) + 2\log p}{n}} + \frac{C_{\gamma}K_{n,p}^2(\log(2n))^{2/\gamma}(\log\left(\frac{3}{\alpha}\right) + 2\log p)^{2/\gamma}}{n},
\]
where $C_{\gamma}$ is a positive universal constant that grows at the rate of $(1/\gamma)^{1/\gamma}$ as $\gamma\downarrow 0$ and
\[
A_{n,p}^2 := \max_{1\le j\le k\le p+1}\, \frac{1}{n}\sum_{i=1}^n \mbox{Var}\left(Z_i(j)Z_i(k)\right).
\]
\end{lem}
\begin{proof}
See Theorem 4.1 of \cite{KuchAbhi17}. A similar result holds for $\gamma > 2$ (the case in which the random variables have tails lighter than the Gaussian). See Theorem 3.4 of \cite{KuchAbhi17} for a result in this direction.
\end{proof}
The confidence regions $\hat{\mathcal{R}}_{n,M}^{\dagger}$ are simple parallelepipeds and can be seen as linear transformations of $\norm{\cdot}_{\infty}$-norm balls. Hence, their Lebesgue measures can be computed exactly. Since the confidence regions are valid over a large number of models, we present a relative Lebesgue measure result uniform over a set of models. For $A\subseteq\mathbb{R}^q$ with $q\ge 1$, let $\Leb(A)$ denote the Lebesgue measure of $A$ with the measure supported on $\mathbb{R}^q$. For convenience, we do not use different notations for the Lebesgue measure for different $q\ge 1$.
\begin{prop}\label{lem:LebesgueMeasure}
For any $k\ge 1$ such that assumption \ref{eq:UniformConsis} are satisfied, the uniform relative Lebesgue measure result holds:
\begin{equation}\label{eq:FistBoundLebesgueMeasure}
\sup_{M\in\mathcal{M}_p(k)}\frac{\Leb\left(\hat{\mathcal{R}}_{n,M}^{\dagger}\right)\Lambda_n^{|M|}(k)}{(C_{n}^{\Gamma}(\alpha) + C_{n}^{\Sigma}(\alpha)\norm{\beta_{n,M}}_1)^{|M|}} = O_p(1).
\end{equation}
Hence, it can be said that $\Leb(\hat{\mathcal{R}}_{n,M}^{\dagger}) = O_p(\mathcal{D}_{n}^{\Gamma} + \mathcal{D}_{n}^{\Sigma}\norm{\beta_{n,M}}_1)^{|M|}$ uniformly for $M\in\mathcal{M}_p(k)$ if $\Lambda_n^{-1}(k) = O(1)$. Moreover, additionally under the setting of Lemma \ref{lem:RateD1nD2n},
\begin{equation}\label{eq:LebesgueRandomX}
\Leb\left(\hat{\mathcal{R}}_{n,M}^{\dagger}\right) = O_p\left(\sqrt{\frac{|M|\log p}{n}}\right)^{|M|}\quad\mbox{uniformly for }M\in\mathcal{M}_p(k),
\end{equation}
if $p$ and $n$ satisfy
\begin{equation}\label{eq:pNotIncrease}
(\log p)^{2/\alpha}(\log n)^{2/\alpha - 1/2} = o(n^{1/2}).
\end{equation}
\end{prop}
\begin{proof}
See Appendix \ref{app:LebesgueMeasure} for a detailed proof.
\end{proof}
\begin{rem}{ (Is the rate optimal?)}
Even though the problem of post-selection inference is studied from various perspectives as discussed in Section~\ref{sec:problem_formulation}, we do not know of a result regarding the optimal size of confidence regions in the post-selection problem. The following argument hints that the rate derived in~\eqref{eq:LebesgueRandomX} is indeed optimal. Since by Theorem~\ref{thm:Uniform} shows simultaneous inference has to be solved for post-selection guarantees, we need to infer about the set of ``parameters'' or functionals
\[
\{\beta_{n,M}(j):\,M\in\mathcal{M}_p(k)\}.
\]
The total number of functionals here is given by
\[
\sum_{\ell = 1}^{k} \binom{p}{\ell}\ell \le k\sum_{\ell = 1}^k \binom{p}{\ell} \le k\sum_{\ell = 1}^k\frac{p^{\ell}}{\ell!} \le k\sum_{\ell = 1}^k \frac{k^{\ell}}{\ell!}\left(\frac{p}{k}\right)^{\ell} \le k\left(\frac{ep}{k}\right)^k \le \left(\frac{2ep}{k}\right)^k.
\]
Even assuming $\sqrt{n}(\hat{\beta}_{n,M}(j) - \beta_{n,M}(j))$ is exactly normal for all $M\in\mathcal{M}_p(k)$ and $j\in M$, we get that
\begin{equation}\label{eq:FullSupremum}
\max_{M\in\mathcal{M}_p(k)}\,\left|\frac{\sqrt{n}\left(\hat{\beta}_{n,M}(j) - \beta_{n,M}(j)\right)}{\sigma_{n,M}(j)}\right| = O_p\left(\sqrt{{k\log(ep/k)}}\right).
\end{equation}
See, for example, Equation (4.3.1) of \cite{DeLaPena99} and the discussion following. Here $\sigma_{n,M}(j)$ represents the variance of $\sqrt{n}(\hat{\beta}_{n,M}(j) - \beta_{n,M}(j))$. Note that the normality assumption implies that
\[
\norm{\sqrt{n}\left(\hat{\beta}_{n,M}(j) - \beta_{n,M}(j)\right)}_{\psi_2} < \infty,
\]
which is enough to apply Equation (4.3.1) of~\cite{DeLaPena99}.

It is possible to get a bound sharper than~\eqref{eq:FullSupremum} with model size dependent scaling. For instance, applying Proposition 4.3.1 of \cite{DeLaPena99}, we get
\begin{equation}\label{eq:SeparatedMaximum}
\max_{1\le \ell \le k}\frac{1}{\sqrt{\log(1 + \ell)}}\max_{M\in\mathcal{M}_p(\ell)\cap\mathcal{M}^c_p(\ell - 1)}\frac{\left|\sqrt{n}\left(\hat{\beta}_{n,M}(j) - \beta_{n,M}(j)\right)/\sigma_{n,M}(j)\right|}{\sqrt{\ell\log(ep/\ell)}} = O_p(1).
\end{equation}
See Appendix~\ref{app:SeparatedMaximum} for a precise statement and proof. This hints that for any model $M$, the confidence region for $\beta_{n,M}$ in the context of simultaneous inference has Lebesgue measure of order $(\sqrt{|M|\log p/n})^{|M|}$. Note that the arguments above are all upper bounds and so they do not prove a lower bound for the Lebesgue measure. This suggests that the Lebesgue measure of our confidence region $\hat{\mathcal{R}}_{n,M}^{\dagger}$ in~\eqref{eq:FirstAsym} is of optimal rate, in general.
\end{rem}

\subsection{Confidence Regions under Fixed Covariates} \label{rem:FixedX}

Since most of the post-selection inference literature as reviewed in Section \ref{sec:Notation} deals with the case of fixed covariates, it is of particular interest to understand how our confidence regions behave in this case. In our framework we can interpret fixed covariates as having point mass distributions at the observed value $X_i$, hence:
\[
\Sigma_n = \frac{1}{n}\sum_{i=1}^n \mathbb{E}\left[X_iX_i^{\top}\right] = \frac{1}{n}\sum_{i=1}^n X_iX_i^{\top} = \hat{\Sigma}_n.
\]
Therefore, $\mathcal{D}_{n}^{\Sigma} = \norm{\hat{\Sigma}_n - \Sigma_n}_{\infty} = 0$ and so, $C_2(\alpha) = 0$. Also, note that in this case
\[
\beta_{n,M} = \left(\frac{1}{n}\sum_{i=1}^n X_i(M)X_i^{\top}(M)\right)^{-1}\left(\frac{1}{n}\sum_{i=1}^n X_i(M)\mathbb{E}\left[Y_i\right]\right).
\]
Hence, in case of fixed covariates,
\begin{align*}
\hat{\mathcal{R}}_{n,M} &= \hat{\mathcal{R}}_{n,M}^{\dagger} = \left\{\norm{\hat\Sigma_n(M)\left\{\hat{\beta}_{n,M} - \beta_{n,M}\right\}}_{\infty} \le C_{n}^{\Gamma}(\alpha)\right\}.
\end{align*}
Note that under fixed covariates assumption \ref{eq:UniformConsis} is trivially satisfied since $\mathcal{D}_{n}^{\Sigma} = 0$. Thus by Theorem \ref{thm:Appr1.2} (or \ref{thm:PoSIX}), finite sample valid post-selection inference holds for all model sizes in case of fixed covariates under no model or distributional assumptions as were required in \cite{Berk13}.

A nice feature of the methodology proposed in \cite{Berk13} is that the inference is tight in the sense there exists a model selection procedure such that the post-selection confidence interval has coverage exactly $1 - \alpha$. Even though the confidence region $\hat{\mathcal{R}}_{n,M}$ is derived under a more general framework, this tightness holds in this generality. This can be easily seen by noting that
\begin{align*}
\sup_{M\in\mathcal{M}_p(p)}\norm{\Sigma_n(M)\left\{\hat{\beta}_{n,M} - \beta_{n,M}\right\}}_{\infty} &= \sup_{M\in\mathcal{M}_p(p)}\left|\frac{1}{n}\sum_{i=1}^n X_i(M)(Y_i - \mathbb{E}\left[Y_i\right])\right|\\ &= \sup_{1\le j\le p}\left|\frac{1}{n}\sum_{i=1}^n X_i(j)(Y_i - \mathbb{E}\left[Y_i\right])\right| = \mathcal{D}_{n}^{\Gamma}.
\end{align*}
Take $\hat{M} = \{\hat{j}\}$, where
\[
\hat{j}\in\argmax_{1\le j\le p}\,\left|\frac{1}{n}\sum_{i=1}^n X_i(j)(Y_i - \mathbb{E}\left[Y_i\right])\right|.
\]
For this random model $\hat{M}$, the coverage of $\hat{\mathcal{R}}_{n,\hat{M}}$ is exactly equal to $(1 - \alpha)$.

\subsubsection[Fixed design and comparison with Berk et al. (2013)]{Lebesgue Measure and Comparison with \cite{Berk13}}

The rate bound~\eqref{eq:LebesgueRandomX} of Lemma~\ref{lem:LebesgueMeasure} is written explicitly for general random covariates. As shown in Remark~\ref{rem:FixedX}, under the assumption of fixed covariates, $C_n^{\Sigma}(\alpha) = 0$ and $\hat{\mathcal{R}}_{n,M} = \hat{\mathcal{R}}_{n,M}^{\dagger}$. So, from the proof of Lemma~\ref{lem:LebesgueMeasure}, we get,
\[
\mathbf{Leb}\left(\hat{\mathcal{R}}_{n,M}\right) \le |\Sigma_n(M)|^{-1}\left(C_n^{\Gamma}(\alpha)\right)^{|M|},\quad\mbox{for all}\quad M\in\mathcal{M}_p(p).
\]
Under the setting of Lemma~\ref{lem:RateD1nD2n}, it follows that
\begin{equation}\label{eq:LebesgueFixedX}
\mathbf{Leb}\left(\hat{\mathcal{R}}_{n,M}\right) = O_p\left(|\Sigma_n(M)|^{-1}\right)\left(\sqrt{\frac{\log p}{n}}\right)^{|M|}.
\end{equation}
Clearly, this is much smaller than the size shown in~\eqref{eq:LebesgueRandomX} for general random covariates. One possible explanation for this discrepancy between fixed and random covariates is as follows: The confidence regions $\hat{\mathcal{R}}_{n,M}$~\eqref{eq:FirstFinite} and $\hat{\mathcal{R}}_{n,M}^{\dagger}$~\eqref{eq:FirstAsym} are written in terms of
\[
\hat{\Sigma}_n(M)\left(\hat{\beta}_{n,M} - \beta_{n,M}\right).
\]
But in case of fixed covariates
\begin{equation}\label{eq:FixedXEstimating}
\hat{\Sigma}_n(M)\beta_{n,M} = \Gamma_n(M).
\end{equation}
So, even though the confidence regions are written for $\beta_{n,M}$, they can be thought of as confidence regions for the population ``parameter'' or functional $\Gamma_n(M)$. Also note that over all models $M\in\mathcal{M}_p(p)$, the set of all functionals $\Gamma_n(M)$ can be inferred just based on $\Gamma_n\in\mathbb{R}^p$. Since this is a $p$-dimensional functional, a confidence region with length $\sqrt{\log p/n}$ on each coordinate can be constructed. This explains why the smaller size in~\eqref{eq:LebesgueFixedX} is possible. In case of random covariates, \eqref{eq:FixedXEstimating} is not true and the randomness due to the covariates brings in some error.

It is striking and somewhat surprising that the smaller size~\eqref{eq:LebesgueFixedX} is possible. In our construction it is not just possible, the confidence region can be computed in polynomial time using bootstrap discussed in Section~\ref{sec:Comp}. The other post-selection methods that can be used in this fixed covariate setting are those of~\cite{Berk13} and \cite{Bac16}. The confidence regions in both these works are based on the quantiles of the statistic
\begin{equation}\label{eq:MaxTStat}
\max_{M\in\mathcal{M}_p(k)}\left|\frac{\sqrt{n}\left(\hat{\beta}_{n,M}(j) - \beta_{n,M}(j)\right)}{\sigma_{n,M}(j)}\right|,
\end{equation}
for some ``variance'' $\sigma_{n,M}(j)$ (The choices of this quantity differ between the works. For simplicity, we assume this quantity is known.) Based on the ``max-$|t|$'' statistic~\eqref{eq:MaxTStat}, a confidence region for $\beta_{n,M}$ is
\[
\hat{\mathcal{R}}_{n,M}^{\mathtt{max-t}} := \left\{\theta\in\mathbb{R}^{|M|}:\,\max_{1\le j\le |M|}\left|\frac{\sqrt{n}\left(\hat{\beta}_{n,M}(j) - \theta(j)\right)}{\sigma_{n,M}(j)}\right| \le C_{n,k}(\alpha)\right\},
\]
where $C_{n,k}(\alpha)$ is the quantile of the max-$|t|$ statistic. Under fixed covariates and Gaussian response, $\sqrt{n}\left(\hat{\beta}_{n,M} - \beta_{n,M}\right)$ is normally distributed. As shown in~\eqref{eq:FullSupremum}, the max-$|t|$ statistic~\eqref{eq:MaxTStat} can be of the order $\sqrt{k\log(ep/k)}$. This implies that $C_{n,k}(\alpha)$ can be of the order $\sqrt{k\log(ep/k)}$ and so, the Lebesgue measure of the confidence region $\hat{\mathcal{R}}_{n,M}^{\mathtt{max-t}}$ satisfies
\begin{equation}\label{eq:GeneralPoSIConstant}
\Leb\left(\hat{\mathcal{R}}_{n,M}^{\mathtt{max-t}}\right)= O_p\left(1\right)\left(\sqrt{\frac{k\log p}{n}}\right)^{|M|}\quad\mbox{uniformly over all}\quad M\in\mathcal{M}_p(k).
\end{equation}
This shows that the confidence region $\hat{\mathcal{R}}_{n,M}^{\mathtt{max-t}}$ is worse than $\hat{\mathcal{R}}_{n,M}^{\dagger}$ in at least two aspects. Firstly, the size of the confidence region has an additional factor $\sqrt{k}$ that makes the region huge in comparison. Secondly, the Lebesgue measure does not scale with model size $|M|$. For example, after searching over the set of models $\mathcal{M}_p(k)$, if the analyst settles on a (random) model of size $1$, then the post-selection confidence region $\hat{\mathcal{R}}_{n,M}^{\mathtt{max-t}}$ has a size that still scales with $k$. In sharp contrast, our confidence region $\hat{\mathcal{R}}_{n,M}^{\dagger}$, even in the random design case, has size scaling only with the model $M$ (and does not depend on the largest model considered in selection process).


\subsubsection{Fixed Covariates with the Restricted Isometry Property (RIP)} \label{sec:RIP}

The rate bound~\eqref{eq:GeneralPoSIConstant} is derived using the fact that $C_{n,k}(\alpha)$ can in general be of the order $\sqrt{k\log(ep/k)}$. Under orthogonal designs ($\hat{\Sigma}_n = I_p$, the identity matrix in $\mathbb{R}^{p\times p}$), \cite{Berk13} proved that $C_{n,k}(\alpha) = O(\sqrt{\log p})$, and so the size of the region $\hat{\mathcal{R}}_{n,M}^{\mathtt{max-t}}$ matches that of our confidence region. Since the construction of~\cite{Berk13} is based on normality, the exact size of the confidence region $\hat{\mathcal{R}}_{n,M}^{\mathtt{max-t}}$ could be better than the region $\hat{\mathcal{R}}_{n,M}^{\dagger}$. It is also interesting to note under orthogonal design $\hat{\mathcal{R}}_{n,M}^{\dagger}$ provides a rectangle with sides parallel to the coordinate axis and so is of the same shape as that of $\hat{\mathcal{R}}_{n,M}^{\mathtt{max-t}}$. Recently, \cite{Bachoc18} showed that the orthogonal design restriction can be relaxed to RIP. A symmetric matrix $A\in\mathbb{R}^{p\times p}$ is said to satisfy RIP of order $k$ with RIP constant $\delta$ if for all $M\in\mathcal{M}_p(k)$ and for all $\theta\in\mathbb{R}^{|M|}$,
\[
(1 - \delta)\norm{\theta}^2 \le \theta^{\top}A(M)\theta \le (1 + \delta)\norm{\theta}^2.
\]
This is equivalent to
\begin{equation}\label{eq:RIPDef}
\max_{|M|\le k}\norm{A(M) - I_{|M|}}_{op} \le \delta,
\end{equation}
where $\norm{\cdot}_{op}$ denotes the operator norm. So, $\hat{\Sigma}_n$ satisfying RIP implies that all $k$ subset covariates are nearly orthogonal. Theorem 3.3 of \cite{Bachoc18} proves that for fixed covariates and Gaussian response,
\[
C_{n,k}(\alpha) = O\left(\sqrt{\frac{\log p}{n}} + \delta c(\delta)\sqrt{\frac{k\log(ep/k)}{n}}\right),
\]
under the assumption that $\hat{\Sigma}_n$ is RIP of order $k$. Here $c(\delta)$ is an increasing non-negative function, satisfying $c(\delta)\to1$ as $\delta\to 0$. So, under the RIP condition with $\delta\sqrt{k}\to 0$, the Lebesgue measure of the confidence region $\hat{\mathcal{R}}_{n,M}^{\mathtt{max-t}}$ matches again with that of our confidence region $\hat{\mathcal{R}}_{n,M}^{\dagger}$. It is also interesting to note that under RIP condition for $\hat{\Sigma}_n$ with $\delta\to 0$, the confidence region $\hat{\mathcal{R}}_{n,M}^{\dagger}$ provides a parallelepiped with sides near parallel to the coordinate axis. More strikingly, the following result holds for fixed covariates:
\begin{prop}\label{prop:ConfUnderRIP}
Define the confidence region
\[
\hat{\mathcal{R}}_{n,M}^{\mathtt{RIP}} := \left\{\theta\in\mathbb{R}^{|M|}:\,\norm{\hat{\beta}_{n,M} - \theta}_{\infty} \le C_n^{\Gamma}(\alpha)\right\}.
\]
If, for any $1\le k\le p$, the matrix $\hat{\Sigma}_n$ satisfies the RIP condition of order $k$ with RIP constant $\delta$ and $\delta\sqrt{k} = o(1)$ as $n\to\infty$, then
\[
\liminf_{n\to\infty}\,\mathbb{P}\left(\bigcap_{M\in\mathcal{M}_p(k)}\left\{\beta_{n,M}\in\hat{\mathcal{R}}_{n,M}^{\mathtt{RIP}}\right\}\right) \ge 1 - \alpha.
\]
\end{prop}
\begin{proof}
From the proof of Theorem~\ref{thm:Appr1.2}, we know that for all $M\in\mathcal{M}_p(k)$,
\[
\norm{\hat{\Sigma}_n(M)\left(\hat{\beta}_{n,M} - \beta_{n,M}\right)}_{\infty} \le \mathcal{D}_{n}^{\Gamma}.
\]
Observe that
\begin{align}
\norm{\hat{\beta}_{n,M} - \beta_{n,M}}_{\infty} &\le \norm{\hat{\Sigma}_n(M)\left(\hat{\beta}_{n,M} - \beta_{n,M}\right)}_{\infty} + \norm{\left(\hat{\Sigma}_n(M) - I_{|M|}\right)\left(\hat{\beta}_{n,M} - \beta_{n,M}\right)}_{\infty}\nonumber\\
&\le \norm{\hat{\Sigma}_n(M)\left(\hat{\beta}_{n,M} - \beta_{n,M}\right)}_{\infty} + \norm{\left(\hat{\Sigma}_n(M) - I_{|M|}\right)\left(\hat{\beta}_{n,M} - \beta_{n,M}\right)}_{2}\nonumber\\
&\le \norm{\hat{\Sigma}_n(M)\left(\hat{\beta}_{n,M} - \beta_{n,M}\right)}_{\infty} + \delta\norm{\hat{\beta}_{n,M} - \beta_{n,M}}_2\nonumber\\
&\le \mathcal{D}_{n}^{\Gamma} + \delta\norm{\hat{\beta}_{n,M} - \beta_{n,M}}_2.\label{eq:LInfinityBound}
\end{align}
From Remark 4.3 of \cite{Uniform:Kuch18}, we get that
\begin{equation}\label{eq:Bachoc33}
\sup_{M\in\mathcal{M}_p(k)}\norm{\hat{\beta}_{n,M} - \beta_{n,M}}_2 \le \frac{\sqrt{k}\mathcal{D}_n^{\Gamma}}{\Lambda_n(k)}.
\end{equation}
(Note that in the notation of~\cite{Uniform:Kuch18}, $\mathcal{D}_n^{\Gamma}$ is different and can be bounded as shown in Proposition 3.1 there the bound above holds.) Therefore, combining~\eqref{eq:LInfinityBound} and~\eqref{eq:Bachoc33}, we get that for all $M\in\mathcal{M}_p(k)$,
\[
\norm{\hat{\beta}_{n,M} - \beta_{n,M}}_{\infty} \le \mathcal{D}_n^{\Gamma}\left(1 + \frac{\delta\sqrt{k}}{\Lambda_n(k)}\right).
\]
From the RIP property~\eqref{eq:RIPDef}, $\Lambda_n(k) \ge 1 - \delta$ and so, for all $M\in\mathcal{M}_p(k)$,
\[
\norm{\hat{\beta}_{n,M} - \beta_{n,M}}_{\infty} \le \mathcal{D}_n^{\Gamma}\left(1 + \frac{\delta\sqrt{k}}{(1 - \delta)}\right).
\]
Therefore, under $\delta\sqrt{k}\to 0$ and using the definition of $C_n^{\Gamma}(\alpha)$,
\[
\liminf_{n\to\infty}\,\mathbb{P}\left(\bigcap_{M\in\mathcal{M}_p(k)}\left\{\beta_{n,M} \in \hat{\mathcal{R}}_{n,M}^{\mathtt{RIP}}\right\}\right) \ge 1 - \alpha.
\]
This completes the proof.
\end{proof}
\begin{rem}{ (RIP is Restrictive)}
The Restricted Isometry Property is a well-known condition in high-dimensional linear regression literature and is also known to be a very restrictive condition. It implies a requirement of near orthogonal covariate subsets, which is often not justified in practice.
\end{rem}
\begin{rem}{ (Generalization of the Result of~\cite{Bachoc18})}
Theorem 3.3 of \cite{Bachoc18} proves a bound on the expectation of $\sup\{\|\hat{\beta}_{n,M} - \beta_{n,M}\|_{\infty}:\,M\in\mathcal{M}_p(k)\}$ for fixed covariates and Gaussian response. Inequality~\eqref{eq:Bachoc33} above proves a deterministic inequality on this supremum quantity. This deterministic inequality along with Lemma~\ref{lem:RateD1nD2n} proves the rate bound in a more general setting.
\end{rem}



\section{Computation by Multiplier Bootstrap}\label{sec:Comp}
All the confidence regions defined in the previous section (and the ones to be defined in the forthcoming sections) depend only on the available data except for the (joint) quantiles $C_{n}^{\Gamma}(\alpha)$ and $C_{n}^{\Sigma}(\alpha)$. Computation or estimation of joint bivariate quantiles $C_{n}^{\Gamma}(\alpha)$ and $C_{n}^{\Sigma}(\alpha)$ is the most important component of an application of approach 1 for valid post-selection inference. In this section, we apply the high-dimensional central limit theorem and multiplier bootstrap for estimating these quantiles. We note that either a classical bootstrap or the recently popularized method of multiplier bootstrap works for estimating these joint quantiles in the setting described in Lemma \ref{lem:RateD1nD2n}. See \cite{Chern17} and \cite{Zhang14} for a detailed discussion. For simplicity, we will only describe the method of multiplier bootstrap for the case of independent random vectors. The discussion here applies the central limit theorem and multiplier bootstrap result proved in Appendix \ref{app:HDCLT}. And we refer to \cite{Zhang14} for the case of dependent settings described in Appendix \ref{app:Dependent}.

Define vectors $W_i\in\mathbb{R}^q$ for $1\le i\le n$ containing
\begin{equation}\label{eq:WiDefinition}
\left(\left\{X_i(j)Y_i\right\}, 1\le j\le p;\; \left\{X_i(l)X_i(m)\right\}, 1\le l \le m\le p\right),
\end{equation}
with
\[
q = 2p + \frac{p(p-1)}{2} = O(p^2).
\]
As shown in Equation \eqref{eq:RectangleRepresentation} in Appendix~\ref{app:HDCLT}, for any $t_1, t_2\in\mathbb{R}^+\cup\{0\}$, the set
\[
\{\mathcal{D}_{n}^{\Gamma} \le t_1, \mathcal{D}_{n}^{\Sigma} \le t_2\},
\]
can be written as a rectangle in terms of
\[
S_n^W := \frac{1}{\sqrt{n}}\sum_{i=1}^n \left\{W_i - \mathbb{E}\left[W_i\right]\right\}.
\]
In the unified framework of linear regression, $(X_i, Y_i)$ are possibly non-identically distributed and so, $\mathbb{E}\left[W_i\right]$ are not all equal. Let $e_1, e_2, \ldots, e_n$ be independent standard normal random variables and define
\[
S_n^{eW} := \frac{1}{\sqrt{n}}\sum_{i=1}^n e_i(W_i - \bar{W}_n),\quad\mbox{where}\quad \bar{W}_n := \frac{1}{n}\sum_{i=1}^n W_i.
\]
Write $S_n^{eW}(\mathbf{I})$ for the first $p$ coordinates of $S_n^{eW}$ and $S_n^{eW}(\mathbf{II})$ for the remaining coordinates of $S_n^{eW}$. The following algorithm gives the pseudo-program for implementing the multiplier bootstrap.
\begin{enumerate}
\item Generate $B_n$ random vectors from $N_n(0, I_n)$, with $I_n$ denoting the identity matrix of dimension $n$. Let these be denoted by $\{e_{i,j}:\, 1\le i\le n, 1\le j\le B_n\}$.
\item Compute the $j$-th replicate of $S_n^{eW}$ as
\[
S_{n,j}^{\star} := \norm{\frac{1}{n}\sum_{i=1}^n e_{i,j}(W_i - \bar{W}_n)}_{\infty},\quad\mbox{for}\quad 1\le j\le B_n.
\]
\item Find any two numbers $(\hat{C}_{1n}^{\Gamma}(\alpha), \hat{C}_{2n}^{\Sigma}(\alpha))$ such that
\[
\frac{1}{B_n}\sum_{i=1}^{B_n} \mathbbm{1}{\left\{\norm{S_{n,j}^{\star}(\mathbf{I})}_{\infty} \le \hat{C}_{1n}^{\Gamma}(\alpha), \norm{S_{n,j}^{\star}(\mathbf{II})}_{\infty} \le \hat{C}_{2n}^{\Sigma}(\alpha)\right\}} \ge 1 - \alpha.
\]
Here $\mathbbm{1}\{A\}$ is the indicator function of a set $A$.
\end{enumerate}
The following theorem proves the validity of multiplier bootstrap under assumption \eqref{eq:MarginalPhi} of Lemma \ref{lem:RateD1nD2n}. Recall the definition of $W_i$ from~\eqref{eq:WiDefinition}. Note that we only prove asymptotic conservativeness instead of consistency which does not hold. See Remark \ref{rem:Inconsistency} in Appendix \ref{app:HDCLT}. This inconsistency can be easily understood by noting that $\mathbb{E}\left[W_i\right]$ is replaced by the average $\bar{W}_n$ which is not a consistent estimator. Define
\[
L_{n,p} := \max_{1\le j\le q}\frac{1}{n}\sum_{i=1}^n \mathbb{E}\left[\left|W_i(j) - \mathbb{E}\left[W_i(j)\right]\right|^3\right].
\]
\begin{thm}\label{thm:BootApplication}
Suppose $(X_i^{\top}, Y_i)^{\top}, 1\le i\le n$ are independent random variables satisfying
\[
\min_{1\le j\le q}\frac{1}{n}\sum_{i=1}^n\mbox{Var}\left(W_i\right) \ge B > 0,
\]
and
\begin{equation}\label{eq:MarginalPhiAlpha}
\max_{1\le i\le n}\max\left\{\max_{1\le j\le p}\norm{X_i(j)}_{\psi_{\gamma}}, \norm{Y_i}_{\psi_{\gamma}}\right\} \le K_{n,p}.
\end{equation}
If $n, p\ge 1$ are such that $$\max\left\{L_{n,p}^{-1}K_{n,p}\left(\log p\right)^{1 + 6/\gamma},\, L_{n,p}^2\log^7p,\, K_{n,p}^{6}\log q,\, K_{n,q}^2(\log p\log n)^{4/\gamma}\right\} = o(n),$$
then the multiplier bootstrap described above provides a conservative inference in the sense that
\[
\lim_{n\to\infty}\inf_{t_1, t_2 \ge 0}\left(\mathbb{P}\left(\mathcal{D}_{n}^{\Gamma} \le t_1, \mathcal{D}_{n}^{\Sigma} \le t_2\right) - \mathbb{P}\left(\norm{S_{n,j}^{eW}(\mathbf{I})}_{\infty} \le t_1, \norm{S_{n,j}^{eW}(\mathbf{II})}_{\infty} \le t_2\big|\mathcal{Z}_n\right)\right) \ge 0,
\]
where $\mathcal{Z}_n := \{(X_i^{\top}, Y_i)^{\top}:\,1\le i\le n\}.$
\end{thm}
\begin{proof}
Theorems \ref{thm:MarginalPhiHDBEBound} and \ref{thm:BootstrapConsistency} (stated in Appendix \ref{app:HDCLT}) apply in the setting above since under assumption \eqref{eq:MarginalPhiAlpha},
\[
\max_{1\le i\le n}\max_{1\le j\le q}\norm{W_i(j)}_{\psi_{\gamma/2}} \le \max_{1\le i\le n}\max\left\{\max_{1\le j\le p}\norm{X_i(j)}_{\psi_{\gamma}}, \norm{Y_i}_{\psi_{\gamma}}\right\}^2 \le K_{n,p}^2.
\]
And the rate restriction on $n$ and $p$ ensure that the bounds in Theorem \ref{thm:MarginalPhiHDBEBound} and \ref{thm:BootstrapConsistency} both converge to zero. See Remark~\ref{rem:Inconsistency} for the conservative property.
\end{proof}
By Theorem \ref{thm:BootApplication}, the estimates $(\hat{C}_{1n}^{\Gamma}(\alpha), \hat{C}_{2n}^{\Sigma}(\alpha))$ are consistent for some quantities that can replace the quantiles $(C_{n}^{\Gamma}(\alpha), C_{n}^{\Sigma}(\alpha))$ of $(\mathcal{D}_{n}^{\Gamma}, \mathcal{D}_{n}^{\Sigma})$ in \eqref{eq:ConfidenceR}.
\begin{rem}\,(Consistency under Identical Distributions)
Under the general framework of just independent random vectors without any assumption on the heterogenity of the distributions, it is impossible to prove consistency as shown in~\cite{Chap2:Kuch18}. The result of~\cite{Chap2:Kuch18} is proved under a much simpler setting but applies here too. If in addition identical distribution of the random vectors is assumed, then it is easy to show from the results of Appendix~\ref{app:HDCLT} that the multiplier bootstrap described above is in fact consistent under the same assumptions of Theorem~\ref{thm:BootApplication}.
\end{rem}


\section{A Generalization for Linear Regression-type Problems}\label{sec:Generalization}
A simple generalization of Theorems \ref{thm:Appr1.2} and \ref{thm:PoSIX} as stated in Theorem \ref{thm:Appr1Gen} allows valid post-selection inference in linear regression-type problems. The importance of this generalization can be seen from Remark \ref{rem:MissingRobust} and the discussion in Section \ref{sec:ProsAndCons}. To describe this generalization, consider the following setting. Let $\hat{\Sigma}_n^{\star}, \Sigma_n^{\star}$ be two $p$-dimensional matrices and $\hat{\Gamma}_n^{\star}, \hat{\Gamma}^{\star}$ be two $p$-dimensional vectors. Consider the error norms
\[
\mathcal{D}_{n}^{\Gamma{\star}} := \norm{\hat{\Gamma}_n^{\star} - \Gamma_n^{\star}}_{\infty}\quad\mbox{and}\quad \mathcal{D}_{n}^{\Sigma{\star}} := \norm{\hat{\Sigma}_n^{\star} - \Sigma_n^{\star}}_{\infty}.
\]
Define for every $M\in\mathcal{M}_p(p)$, the estimator and the corresponding target as
\begin{align*}
\hat{\xi}_{n,M} &:= \argmin_{\theta\in\mathbb{R}^{|M|}}\,\left\{\theta^{\top}\hat{\Sigma}_n^{\star}(M)\theta - 2\theta^{\top}\hat{\Gamma}_n^{\star}(M)\right\},\\
\xi_{n,M} &:= \argmin_{\theta\in\mathcal{R}^{|M|}}\,\left\{\theta^{\top}\Sigma_n^{\star}(M)\theta - 2\theta^{\top}\Gamma_n^{\star}(M)\right\}.
\end{align*}
Consider the confidence regions $\hat{\mathcal{R}}_{n,M}^{\star}$ and $\hat{\mathcal{R}}_{n,M}^{{\star}\dagger}$, analogues to those before, as
\begin{align*}
\hat{\mathcal{R}}_{n,M}^{\star} &:= \left\{\theta\in\mathbb{R}^{|M|}:\,\norm{\hat{\Sigma}_n^{\star}(M)\left(\hat{\xi}_{n,M} - \theta\right)}_{\infty} \le C_{n}^{\Gamma{\star}}(\alpha) + C_{n}^{\Sigma{\star}}(\alpha)\norm{\theta}_1\right\},\\
\hat{\mathcal{R}}_{n,M}^{{\star}\dagger} &:= \left\{\theta\in\mathbb{R}^{|M|}:\,\norm{\hat{\Sigma}_n^{\star}(M)\left(\hat{\xi}_{n,M} - \theta\right)}_{\infty} \le C_{n}^{\Gamma{\star}}(\alpha) + C_{n}^{\Sigma{\star}}(\alpha)\norm{\hat{\xi}_{n,M}}_1\right\}.
\end{align*}
where $C_{n}^{\Gamma\star}(\alpha)$ and $C_{n}^{\Sigma\star}(\alpha)$ are constants (or joint quantiles) that satisfy,
\[
\mathbb{P}\left(\mathcal{D}_{n}^{\Gamma{\star}} \le C_{n}^{\Gamma\star}(\alpha)\quad\mbox{and}\quad \mathcal{D}_{n}^{\Sigma{\star}} \le C_{n}^{\Sigma\star}(\alpha)\right) \ge 1 - \alpha.
\]
Finally, let $\Lambda_n^{\star}(k) = \min\{\lambda_{\min}(\Sigma_n^{\star}(M)):\,M\in \mathcal{M}_p(k)\}$.
\begin{thm}\label{thm:Appr1Gen}
The set of confidence regions $\{\hat{\mathcal{R}}_{n,M}^{\star}:\,M\in\mathcal{M}_p(p)\}$ satisfies
\begin{equation}\label{eq:Simultaneous}
\mathbb{P}\left(\bigcap_{M\in\mathcal{M}_p(p)}\left\{\xi_{n,M} \in \hat{\mathcal{R}}_{n,M}^{\star}\right\}\right) \ge 1 - \alpha,
\end{equation}
and if for any $1\le k\le p$ that satisfies $k\mathcal{D}^{\star}_{2n} = o_p(\Lambda^{\star}_n(k)) = o_p(1)$,
\[
\liminf_{n\to\infty}\,\mathbb{P}\left(\bigcap_{M\in\mathcal{M}_p(k)}\left\{\xi_{n,M} \in \hat{\mathcal{R}}_{n,M}^{{\star}\dagger}\right\}\right) \ge 1 - \alpha.
\]
\end{thm}
\begin{proof}
The proof is exactly the same as for Theorems \ref{thm:Appr1.2} and \ref{thm:PoSIX}. The reader just has to realize that we did not use any structure of $\hat{\Sigma}_n, \hat{\Gamma}_n$ or that they are unbiased estimators of $\Sigma_n, \Gamma_n$ respectively, in the proof there.
\end{proof}
\begin{rem}\label{rem:MissingRobust}
The result in Theorem \ref{thm:Appr1Gen} allows one to deal with the case of missing data or outliers in linear regression setting. In case of missing data or when the data is suspected of containing outliers, it might be more useful to use estimators of $\Sigma_n$ and $\Gamma_n$ that take this concern into account. For the case of missing data/errors-in-covariates/multiplicative noise, see \citet[Examples 1, 2 and 3]{Loh12} and references therein for estimators other than $\hat{\Sigma}_n$ and $\hat{\Gamma}_n$. For the case of outliers either in the classical sense or in the adversarial corruption setting, see \cite{Chen13}. For correct usage of this theorem, it is crucial that the sub-matrix and sub-vector of $\Sigma_n^{\star}$ and $\Gamma_n^{\star}$, respectively are used for sub-models. For example, if we use full covariate imputation in case of missing data, then the sub-model estimator should be based on a sub-matrix of this full covariate imputation. Also, see \citet[pages 11--12]{Uniform:Kuch18} for other settings of applicability.
\end{rem}


\section{Connection to High-dimensional Regression and Other Confidence Regions}\label{sec:HighDim}
The confidence regions $\hat{\mathcal{R}}_{n,M}$ and $\hat{\mathcal{R}}_{n,M}^{\dagger}$ have a very close connection to a well-known estimator in the high-dimensional linear regression literature called the Dantzig Selector proposed by \cite{Can07} and the closely related ones by \cite{Ros10} and \cite{Chen13}. These papers or methods are not related to post-selection inference and were proposed under a linear model assumption. The Dantzig selector estimates $\beta_0\in\mathbb{R}^p$, using observations $(X_i^{\top}, Y_i), 1\le i\le n$ that satisfy $Y_i = X_i^{\top}\beta_0 + \varepsilon_i$ for independent and identically distributed errors $\varepsilon_i$ with a mean zero normal distribution.  \cite{Can07}, like many others, assumed fixed covariates $X_i, 1\le i\le n$. In our notation, the Dantzig selector is defined by the optimization problem
\[
\mbox{minimize}\quad\norm{\beta}_1\quad\mbox{subject to}\quad \norm{\Gamma_n - \Sigma_n\beta}_{\infty}\le \lambda_n,
\]
for some tuning parameter $\lambda_n$ that converges to zero as $n$ increases. To relate this to our confidence regions $\hat{\mathcal{R}}_{n,M}^{\dagger}$ (in \eqref{eq:FirstFinite}), note that for $\beta = \beta_0$ in the constraint set, the quantity inside the norm is $\Sigma_n(\hat{\beta} - \beta_0)$ where $\hat{\beta}$ is any least squares estimator. 
The estimator defined in \cite{Chen13} and \cite{Ros10} resembles
\[
\mbox{minimize}\quad\norm{\beta}_1\quad\mbox{subject to}\quad \norm{\Gamma_n - \Sigma_n\beta}_{\infty} \le \lambda_n + \delta_n\norm{\beta}_1,
\]
for some tuning parameters $\lambda_n$ and $\delta_n$ both converging to zero as $n$ increases. This constraint set corresponds to our confidence regions $\hat{\mathcal{R}}_{n,M}$ in Theorem \ref{thm:Appr1.2}.

The following theorem proves that there exist valid post-selection confidence regions that resemble the objective functions of lasso (\cite{Tibs96}) and sqrt-lasso (\cite{Belloni11}). The proof is deferred to Appendix \ref{app:LassoRegions}. These relations to high-dimensional linear regression literature poses the interesting question: ``is there a more deeper connection between post-selection inference and high-dimensional estimation?''. Other than the results in linear regression, we do not yet have an answer to this interesting question.

Define for every $M\in\mathcal{M}_p(p)$, the confidence regions
\begin{align*}
\mathring{\mathcal{R}}_{n,M} &:= \left\{\theta\in\mathbb{R}^{|M|}:\,\right.\\
&\left.\hat{R}_n(\theta; M) \le \hat{R}_n(\hat{\beta}_{n,M}; M) + 2C_{n}^{\Gamma}(\alpha)\left[\norm{\hat{\beta}_{n,M}}_1 + \norm{\theta}_1\right] + C_{n}^{\Sigma}(\alpha)\left[\norm{\hat{\beta}_{n,M}}_1^2 + \norm{\theta}_1^2\right]\right\},\\
\mathring{\mathcal{R}}_{n,M}^{\dagger} &:= \left\{\theta\in\mathbb{R}^{|M|}:\,\hat{R}_n(\theta; M) \le \hat{R}_n(\hat{\beta}_{n,M}; M) + 4C_{n}^{\Gamma}(\alpha)\norm{\hat{\beta}_{n,M}}_1 + 2C_{n}^{\Sigma}(\alpha)\norm{\hat{\beta}_{n,M}}_1^2\right\},\\
\breve{\mathcal{R}}_{n,M} &:= \left\{\theta\in\mathbb{R}^{|M|}:\,\right.\\
&\quad\left.\hat{R}^{1/2}_n(\theta; M) \le \hat{R}_n^{1/2}(\hat{\beta}_{n,M}; M) + C_n^{1/2}(\alpha)\left(1 + \norm{\theta}_1\right)+ C_n^{1/2}(\alpha)\left(1 + \norm{\hat{\beta}_{n,M}}_1\right)\right\},\\
\breve{\mathcal{R}}_{n,M}^{\dagger} &:= \left\{\theta\in\mathbb{R}^{|M|}:\,\hat{R}^{1/2}_n(\theta; M) \le \hat{R}_n^{1/2}(\hat{\beta}_{n,M}; M) + 2C_n^{1/2}(\alpha)\left(1 + \norm{\hat{\beta}_{n,M}}_1\right)\right\},
\end{align*}
where $\hat{R}_n(\cdot; M)$ is the empirical least squares objective function defined in Equation \eqref{eq:EmpObj} and $C_n(\alpha)$ is the $(1 - \alpha)$-upper quantile of $\max\{\mathcal{D}_{n}^{\Gamma}, \mathcal{D}_{n}^{\Sigma}\}$.
\begin{thm}\label{thm:PoSILasso}
For any $n\ge 1, p\ge 1$, the following simultaneous inference guarantee holds:
\begin{align}
\mathbb{P}\left(\bigcap_{M\in\mathcal{M}_p(p)}\left\{\beta_{n,M}\in\mathring{\mathcal{R}}_{n,M}\right\}\right) &\ge 1 - \alpha,\label{eq:LassoFinite}\\
\mathbb{P}\left(\bigcap_{M\in\mathcal{M}_p(p)}\left\{\beta_{n,M}\in\breve{\mathcal{R}}_{n,M}\right\}\right) &\ge 1 - \alpha,\label{eq:SqrtLassoFinite}
\end{align}
and for any $1\le k\le p$ satisfying \ref{eq:UniformConsis}, we have
\begin{align}
\liminf_{n\to\infty}\,\mathbb{P}\left(\bigcap_{M\in\mathcal{M}_p(p)}\left\{\beta_{n,M}\in\mathring{\mathcal{R}}_{n,M}^{\dagger}\right\}\right) &\ge 1 - \alpha,\label{eq:LassoAsym}\\
\liminf_{n\to\infty}\,\mathbb{P}\left(\bigcap_{M\in\mathcal{M}_p(p)}\left\{\beta_{n,M}\in\breve{\mathcal{R}}_{n,M}^{\dagger}\right\}\right) &\ge 1 - \alpha,\label{eq:SqrtLassoAsym}
\end{align}
\end{thm}
\begin{rem}{ (Intersection of Confidence Regions)}\label{rem:Intersect}
All our confidence regions are based on deterministic inequalities as mentioned before. This implies that the intersection of the confidence regions $\hat{\mathcal{R}}_{n,M}, \hat{\mathcal{R}}_{n,M}^{\dagger}$ and $\mathring{\mathcal{R}}_{n,M}$ provides a valid simultaneous and post-selection inference. That means, for any $1\le k\le p$ such that \ref{eq:UniformConsis} holds,
\begin{equation}\label{eq:Intersect}
\liminf_{n\to\infty}\,\mathbb{P}\left(\bigcap_{M\in\mathcal{M}_p(k)}\left\{\hat{\mathcal{R}}_{n,M}\cap\hat{\mathcal{R}}_{n,M}^{\dagger}\cap \mathring{\mathcal{R}}_{n,M}\right\}\right) \ge 1 - \alpha.
\end{equation}
To prove this, let $\hat{\mathcal{C}}_{n,M}, \hat{\mathcal{C}}_{n,M}^{\dagger}$ and $\mathring{\mathcal{C}}_{n,M}$ represent the confidence sets $\hat{\mathcal{R}}_{n,M}, \hat{\mathcal{R}}_{n,M}^{\dagger}$ and $\mathring{\mathcal{R}}_{n,M}$ with $(C_{n}^{\Gamma}(\alpha), C_{n}^{\Sigma}(\alpha))$ replaced by $(\mathcal{D}_{n}^{\Gamma}, \mathcal{D}_{n}^{\Sigma})$. From the proofs of Theorems \ref{thm:Appr1.2}, \ref{thm:PoSIX} and \ref{thm:PoSILasso}, it is clear that
\[
\liminf_{n\to\infty}\,\mathbb{P}\left(\bigcap_{M\in\mathcal{M}_p(k)}\left\{\hat{\mathcal{C}}_{n,M}\cap\hat{\mathcal{C}}_{n,M}^{\dagger}\cap\mathring{\mathcal{C}}_{n,M}\right\}\right) = 1.
\]
So by the definition of $(C_{n}^{\Gamma}(\alpha), C_{n}^{\Sigma}(\alpha))$ \eqref{eq:ConfidenceR}, the result of \eqref{eq:Intersect} follows. Provably the intersection of confidence regions is smaller. By the same argument it is possible to include the confidence regions $\mathring{\mathcal{R}}_{n,M}^{\dagger}, \breve{\mathcal{R}}_{n,M},$ and $\breve{\mathcal{R}}_{n,M}^{\dagger}$ in the intersection .
\end{rem}
\begin{rem}{ (Usefulness of Lasso-based Regions)}
The confidence regions discussed in this section are given solely for the purpose of illustrating and making solid the connection between post-selection inference and high-dimensional linear regression. The shape of all these confidence regions is ellipsoid and have larger volume than the confidence region $\hat{\mathcal{R}}_{n,M}^{\dagger}$ in terms of the rate. This result is not presented here but is not difficult to prove. This rate comparison is only asymptotic and the intersection argument presented in Remark \ref{rem:Intersect} might still be useful in finite samples.
\end{rem}


\section{Discussion of the Current Approach}\label{sec:ProsAndCons}

The confidence regions $\hat{\mathcal{R}}_{n,M}$ and $\hat{\mathcal{R}}_{n,M}^{\dagger}$ constitute what we call approach 1. Various advantages and disadvantages of this approach are discussed in this section. Some of these comments also apply to the confidence regions mentioned in Theorem \ref{thm:Appr1Gen}.

The following are some of the advantages of this approach. The confidence regions are asymptotically valid for post-selection inference. This is the first work that provides valid post-selection inference in this generality. The confidence region for any model $M$ depend only on the joint quantiles $C_{n}^{\Gamma}(\alpha), C_{n}^{\Sigma}(\alpha)$ and the least squares linear regression estimator corresponding to the model $M$, $\hat{\beta}_{n,M}$. So, the computational complexity of these confidence regions is no more than a multiple of the computational complexity of $\hat{\beta}_{n,M}$. Computation of $C_{n}^{\Gamma}(\alpha), C_{n}^{\Sigma}(\alpha)$ takes no more than a linear function of $p$ operations, as shown in Section \ref{sec:Comp}. This computational complexity is in sharp contrast to the valid post-selection inference method proposed by \cite{Berk13} or \cite{Bac16} which requires essentially solving for the least squares estimators of all the models for a confidence region with some model $M$. Therefore, implementation of their procedure is NP-hard, in general. The Lebesgue measure of the confidence regions $\hat{\mathcal{R}}_{n,M}^{\dagger}$ converges to zero at a rate that is the minimax rate in high-dimensional linear regression literature. So, we suspect this might be the optimal rate here too but at present we do not have a proof or even an optimality framework. Note that the volume of the confidence region for model $M$ is computed with respect to the Lebesgue on $\mathbb{R}^{|M|}.$

There is one more advantage which might not seem like one at first glance. The confidence region for $\beta_{n,M}$ for a particular model does not require information on how many models are being used for model selection. The volume of the confidence region for $\beta_{n,M}$ depends only on the features of the model $M$ except for the quantiles. This implies that the confidence regions $\hat{\mathcal{R}}_{n,M}^{\dagger}, M\in\mathcal{M}_p(k)$ can often have much smaller volumes than the ones produced using the approach of \cite{Berk13}.

There are some disadvantages and some irking factors associated with this approach. Firstly, notice that the confidence regions are not invariant under linear transformations of the observations as briefed in Remark \ref{rem:Invariance}. Most methods in high-dimensional linear regression procedures that induce sparsity also share this feature. Even from a naive point of view, invariance under change of units for all variables involved is crucial for interpretation. This translates to invariance under diagonal linear transformations of the observations. Normalizing all the variables involved to have a unit standard deviation is a commonly suggested method to attain invariance under diagonal transformations. Formally, this means one should use
\[
{X}_i^* = \left(\frac{X_i(1) - \bar{X}(1)}{s_n(1)}, \ldots, \frac{X_i(p) - \bar{X}_i(p)}{s_n(p)}\right),\quad {Y}_i^* = \frac{Y_i - \bar{Y}}{s_n(0)},
\]
in place of $(X_i, Y_i), 1\le i\le n$, where for $1\le j\le p$,
\[
\bar{X}(j) = \frac{1}{n}\sum_{i=1}^n X_i(j),\quad\mbox{and}\quad s_n^2(j) = \frac{1}{n}\sum_{i=1}^n \left[X_i(j) - \bar{X}(j)\right]^2,
\]
and
\[
\bar{Y} = \frac{1}{n}\sum_{i=1}^n Y_i,\quad\mbox{and}\quad s_{n}^2(0) = \frac{1}{n}\sum_{i=1}^n \left[Y_i - \bar{Y}\right]^2.
\]
This leads to the matrix and vector,
\[
\hat\Sigma_n^{\star} = \frac{1}{n}\sum_{i=1}^n {X}_i^*{X}_i^{*\top},\quad\mbox{and}\quad \hat\Gamma_n^{\star} = \frac{1}{n}\sum_{i=1}^n {X}_i^*{Y}_i^*.
\]
Note that the observations $({X}_i^*, {Y}_i^*), 1\le i\le n$ are not independent even if we start with independent observations $(X_i, Y_i)$. This is one of the reasons why we did not assume independence for Theorems \ref{thm:Appr1.2}, \ref{thm:PoSIX} and \ref{thm:Appr1Gen}. Of course one needs to prove the rates for the error norms $\mathcal{D}_{n}^{\Gamma\star}$ and $\mathcal{D}_{n}^{\Sigma\star}$ in this case for an application of these results. We leave it to the reader to verify that the rates are exactly the same obtained in Lemma \ref{lem:RateD1nD2n} (one needs to use a Slutsky-type argument). See \cite{Cui16} for a similar derivation. We conjecture that much weaker conditions than listed in Lemma \ref{lem:RateD1nD2n} are enough for those same rates, in particular, exponential moments are not required. See \citet[Theorem 5.3]{Geer14} for a result in this direction. Getting back to invariance under arbitrary linear transformations, we do not know if it is possible come up with a procedure that retains the computational complexity of approach 1 while satisfying this invariance. We conjecture that this is not possible and that there is a strict trade-off between computational efficiency and affine invariance.

Another disadvantage of approach 1 is that it is mostly based on deterministic inequalities. As the reader may have suspected, this might lead to some conservativeness of the method. Note that non-identical distributions of the observations already introduces some conservativeness. The confidence regions $\hat{\mathcal{R}}_{n,M}$ and $\hat{\mathcal{R}}_{n,M}^{\dagger}$ cover $\beta_{n,M}$ with probability (at least) $1 - \alpha$ asymptotically. In particular, these confidence regions provide valid post-selection inference for the full vector $\beta_{n,M}$ instead of each of the coordinates of $\beta_{n,M}$. The region $\hat{\mathcal{R}}_{n,M}^{\dagger}$ is defined by a system of linear inequalities and hence the local inference (or inference on coordinates) for $\beta_{n,M}(j), 1\le j\le |M|$ can be obtained by solving a linear program. However, these can be very conservative for local inference guarantees.

We emphasize before ending this section that the main focus of approach 1 is validity and better computational complexity not optimality. However, optimality holds for our confidence regions as mentioned in Remark \ref{rem:FixedX} for fixed covariates. It should be understood that without validity there is no point in proving any kind of optimality properties about the size of confidence region. 


\section{Conclusions and Future Directions}\label{sec:Conclusions}
In this paper, we have considered a computationally efficient approach to valid post-selection inference in linear regression under arbitrary data-driven method of variable selection. The approach here is very different from the other methodologies available in the literature and is based on the estimating equation of linear regression. At present it is not clear if this approach can be extended to other $M$-estimation problems. Since our confidence regions are based on deterministic inequalities, our results provide valid post-selection inference even under dependence and non-identically distributed random vectors. For this reason, the setting of the current work is the most general available in the literature of post-selection inference. 

In addition to providing several valid confidence regions, we compare the Lebesgue measure of our confidence regions with the ones from~\cite{Berk13} and~\cite{Bac16}. This comparison shows that our confidence regions are much smaller (in terms of volume) in case of fixed (non-stochastic) covariates. In general, the volume of our confidence regions scales with the cardinality of model $\hat{M}$ chosen. This is a feature not available from the works of~\cite{Berk13} and~\cite{Bac16}. Note that the confidence regions from selective inference literature have infinite expected length as shown in~\cite{2018arXiv180301665K}.

An interesting finding of our work is the connection between post-selection confidence regions and high-dimensional sparsity inducing linear regression estimators. If this finding were to hold for other $M$-estimation problems, then computationally efficient valid post-selection confidence regions are possible in general.   
\bibliographystyle{apalike}
\bibliography{AssumpLean}
\appendix

\medskip
\begin{center}
APPENDIX
\end{center}
\section{Proof of Lemma \ref{lem:UniformConsisL1}}\label{app:ProofLemmaL1}

Fix $M\in\mathcal{M}_p(k)$ with $k\mathcal{D}_{n}^{\Sigma} \le \Lambda_n(k)$. Observe that the least squares estimator satisfies
\begin{equation}\label{eq:LeastSquaresEquation}
\hat{\beta}_{n,M} - \beta_{n,M} = \left(\Sigma_n(M)\right)^{-1}\left(\left[\hat{\Gamma}_n(M) - \Gamma_n(M)\right] - \left[\hat{\Sigma}_n(M) - \Sigma_n(M)\right]\beta_{n,M}\right),
\end{equation}
and for all $M\in\mathcal{M}_p(k),$
\begin{align}
\norm{\hat{\Sigma}_n(M) - \Sigma_n(M)}_{op} &\le  \sup_{\substack{\norm{\theta}_0 \le k,\\\norm{\theta}_2 \le 1}}\left|\theta^{\top}\left(\hat{\Sigma}_n - \Sigma_n\right)\theta\right| \le k\norm{\hat{\Sigma}_n - \Sigma_n}_{\infty} = k\mathcal{D}_{n}^{\Sigma}.\label{eq:UniformNonSing}
\end{align}
Thus, for all $M\in\mathcal{M}_p(k)$,
\[
\Lambda_n(k) - k\mathcal{D}_{n}^{\Sigma} \le \norm{\Sigma_n(M)}_{op} - k\mathcal{D}_{n}^{\Sigma} \le \norm{\hat{\Sigma}_n(M)}_{op} \le \norm{\Sigma_n(M)}_{op} + k\mathcal{D}_{n}^{\Sigma}.
\]
Hence, for $k$ satisfying $k\mathcal{D}_{n}^{\Sigma} \le \Lambda_n(k)$,
\begin{align*}
\norm{\hat{\beta}_{n,M} - \beta_{n,M}}_2 &\le \frac{\norm{\hat{\Gamma}_n(M) - \Gamma_n(M)}_2 + \norm{[\hat{\Sigma}_n(M) - \Sigma_n(M)]\beta_{n,M}}_2}{\Lambda_n(k) - k\mathcal{D}_{n}^{\Sigma}}\\
&\le |M|^{1/2}\frac{\norm{\hat{\Gamma}_n(M) - \Gamma_n(M)}_{\infty} + \norm{[\hat{\Sigma}_n(M) - \Sigma_n(M)]\beta_{n,M}}_{\infty}}{\Lambda_n(k) - k\mathcal{D}_{n}^{\Sigma}}\\
&\le \frac{|M|^{1/2}\left(\mathcal{D}_{n}^{\Gamma} + \mathcal{D}_{n}^{\Sigma}\norm{\beta_{n,M}}_1\right)}{\Lambda_n(k) - k\mathcal{D}_{n}^{\Sigma}}.
\end{align*}
Now applying
\[
\norm{\hat{\beta}_{n,M} - \beta_{n,M}}_1 \le {|M|}^{1/2}\norm{\hat{\beta}_{n,M} - \beta_{n,M}}_2,
\]
the result follows.


\section{Proof of Proposition \ref{lem:LebesgueMeasure}}\label{app:LebesgueMeasure}

For any fixed model $M$, the Lebesgue measure of the confidence region is given by
\begin{equation}\label{eq:Lebesgue}
\Leb(\hat{\mathcal{R}}_{M}^{\dagger}) = |\Sigma_n(M)|^{-1}\left(C_{n}^{\Gamma}(\alpha) + C_{n}^{\Sigma}(\alpha)\norm{\hat{\beta}_M}_1\right)^{|M|},
\end{equation}
which converges to zero as $n$ tends to infinity. Here for any matrix $A\in\mathbb{R}^{p\times p}$, $|A|$ denotes the determinant of $A$. This equality follows since the confidence region $\hat{\mathcal{R}}_{M}^{\dagger}$ can be written as
\[
\hat{\mathcal{R}}_{M}^{\dagger} = \left\{\left[\Sigma_n(M)\right]^{-1}(\theta + \hat{\beta}_M):\, \norm{\theta}_{\infty} \le \left(C_{n}^{\Gamma}(\alpha) + C_{n}^{\Sigma}(\alpha)\norm{\hat{\beta}_{n,M}}_1\right)\right\}.
\]
By inequality \eqref{eq:UniformNonSing}, for all $M\in\mathcal{M}_p(k)$
\[
|\Sigma_n(M)|^{-1} \le \left(\Lambda_n(k) - k\mathcal{D}_{n}^{\Sigma}\right)^{-|M|}.
\]
We know that $C_{n}^{\Gamma}(\alpha)$ and $C_{n}^{\Sigma}(\alpha)$ converge to zero at a rate depending on the tails of the joint distribution of $(X_i, Y_i)$. The result now follows from equation \eqref{eq:Lebesgue} and uniform consistency of $\hat{\beta}_{n,M}$ in the $\norm{\cdot}_1$-norm as shown in Lemma \ref{lem:UniformConsisL1} under \ref{eq:UniformConsis}.

To prove the second result, first note that from Lemma \ref{lem:RateD1nD2n},
\[
\max\{C_{n}^{\Gamma}(\alpha), C_{n}^{\Sigma}(\alpha)\} = O\left(\sqrt{\frac{\log p}{n}}\right),
\]
since the second term in the expectation bound in Lemma \ref{lem:RateD1nD2n} is of lower order than the first term under the assumption \eqref{eq:pNotIncrease} of Lemma \ref{lem:LebesgueMeasure}. The result is now proved if we prove that for all $M\in\mathcal{M}_p(k)$,
\begin{equation}\label{eq:BoundCoefficients}
\norm{\beta_{n,M}}_1^2 \le \frac{|M|}{\Lambda_n(k)}\left(\frac{1}{n}\sum_{i=1}^n \mathbb{E}\left[Y_i^2\right]\right).
\end{equation}
By definition of $\beta_{n,M}$ it follows that
\[
0 \le \frac{1}{n}\sum_{i=1}^n \mathbb{E}\left[Y_i^2\right] - \frac{1}{n}\sum_{i=1}^n \beta_{n,M}^{\top}\mathbb{E}\left[X_i(M)X_i^{\top}(M)\right]\beta_{n,M} = \frac{1}{n}\sum_{i=1}^n \mathbb{E}\left[\left(Y_i - X_i^{\top}(M)\beta_{n,M}\right)^2\right].
\]
Therefore, by definition of $\Lambda_n(k)$,
\[
\Lambda_n(k)\norm{\beta_{n,M}}_2^2 \le \left(\frac{1}{n}\sum_{i=1}^n \mathbb{E}\left[Y_i^2\right]\right).
\]
Now using the inequality $\norm{\beta_{n,M}}_1 \le \sqrt{|M|}\norm{\beta_{n,M}}_2$, inequality \eqref{eq:BoundCoefficients} follows.
\section{Proof of~\texorpdfstring{\eqref{eq:SeparatedMaximum}}{blah}}\label{app:SeparatedMaximum}
\begin{prop}\label{prop:SeparatedMaximum}
Suppose there exists a constant $B_{n,k,p}$ and $\gamma > 0$, such that
\[
\sup_{M\in\mathcal{M}_p(k)}\max_{1\le j\le |M|}\norm{\frac{\sqrt{n}\left(\hat{\beta}_{n,M}(j) - \beta_{n,M}(j)\right)}{\sigma_{n,M}(j)}}_{\psi_{\gamma}} \le B_{n,k,p}.
\]
Then
\[
\max_{1\le \ell \le k}\frac{1}{\psi_{\gamma}^{-1}(\ell)}\max_{M\in\mathcal{M}_p(\ell)\cap\mathcal{M}^c_p(\ell - 1)}\frac{\left|\sqrt{n}\left(\hat{\beta}_{n,M}(j) - \beta_{n,M}(j)\right)/\sigma_{n,M}(j)\right|}{\psi_{\gamma}^{-1}\left((2ep/\ell)^{\ell}\right)} = O_p(1)
\]
\end{prop}
\begin{proof}
From the proof of Proposition A.5 of \cite{KuchAbhi17}, we get
\[
\norm{\max_{M\in\mathcal{M}_p(\ell)\cap\mathcal{M}^c_p(\ell - 1)}{\left|\frac{\sqrt{n}\left(\hat{\beta}_{n,M}(j) - \beta_{n,M}(j)\right)}{\sigma_{n,M}(j)}\right|}}_{\psi_{\gamma}} \le \psi_{\gamma}^{-1}\left((ep/\ell)^{\ell}\right)C_{\gamma}B_{n,p,k},
\]
for some constant $C_{\gamma}$ depending only on $\gamma$. Here the fact
\[
\left|\mathcal{M}_p(\ell)\setminus\mathcal{M}_p(\ell - 1)\right| = \binom{p}{\ell} \le \left(\frac{ep}{\ell}\right)^{\ell},
\]
is used. Now take
\[
\xi_{\ell} := \frac{1}{\psi_{\gamma}^{-1}\left((ep/\ell)^{\ell}\right)}\max_{M\in\mathcal{M}_p(\ell)\cap\mathcal{M}^c_p(\ell - 1)}{\left|\frac{\sqrt{n}\left(\hat{\beta}_{n,M}(j) - \beta_{n,M}(j)\right)}{\sigma_{n,M}(j)}\right|},
\]
and apply Proposition 4.3.1 of \cite{DeLaPena99}, to get the result. Also, see Proposition A.7 of \cite{KuchAbhi17} for an alternative proof to Proposition 4.3.1 of \cite{DeLaPena99}. If $\gamma = 2$, then $\psi_\gamma$ corresponds to sub-Gaussian random variables and
\[
\psi_2^{-1}(x) = \sqrt{\log(1 + x)}.
\]
This proves~\eqref{eq:SeparatedMaximum}.
\end{proof}


\section{Proof of Theorem \ref{thm:PoSILasso}}\label{app:LassoRegions}

Only the proof of \eqref{eq:LassoFinite} and \eqref{eq:LassoAsym} is provided and the steps to prove \eqref{eq:SqrtLassoFinite} and \eqref{eq:SqrtLassoAsym} are sketched since the proof is similar.

It is easy to verify that for any $M\subseteq\mathcal{M}_p(p)$ and $\theta\in\mathbb{R}^{|M|}$
\begin{equation}\label{eq:FundamentalIneq}
\left|\theta^{\top}\hat{\Sigma}_n(M)\theta - 2\theta^{\top}\hat{\Gamma}_n(M) - \theta^{\top}\Sigma_n(M)\theta + 2\theta^{\top}\Gamma_n(M)\right| \le \norm{\theta}_1^2\mathcal{D}_{n}^{\Sigma} + 2\norm{\theta}_1\mathcal{D}_{n}^{\Gamma}.
\end{equation}
Therefore, for every $M\in\mathcal{M}_p(p)$,
\begin{align*}
\beta_{n,M}&\hat{\Sigma}_n(M)\beta_{n,M} -2\beta_{n,M}^{\top}\hat{\Gamma}_n(M)\\ &\le \beta_{n,M}\Sigma_n(M)\beta_{n,M} -2\beta_{n,M}^{\top}\Gamma_n(M) + 2\mathcal{D}_{n}^{\Gamma}\norm{\beta_{n,M}}_1 + \mathcal{D}_{n}^{\Sigma}\norm{\beta_{n,M}}_1^2\\
&\le \hat{\beta}_{n,M}\Sigma_n(M)\hat{\beta}_{n,M} -2\hat{\beta}_{n,M}^{\top}\Gamma_n(M) +  2\mathcal{D}_{n}^{\Gamma}\norm{\beta_{n,M}}_1 + \mathcal{D}_{n}^{\Sigma}\norm{\beta_{n,M}}_1^2\\
&\le \hat{\beta}_{n,M}\hat{\Sigma}_n(M)\hat{\beta}_{n,M} -2\hat{\beta}_{n,M}^{\top}\Gamma_n(M) + 2\mathcal{D}_{n}^{\Gamma}\left[\norm{\hat{\beta}_{n,M}}_1 + \norm{\beta_{n,M}}_1\right]\\&\qquad+ \mathcal{D}_{n}^{\Sigma}\left[\norm{\hat{\beta}_{n,M}}_1^2 + \norm{\beta_{n,M}}_1^2\right].
\end{align*}
Here the first inequality follows from inequality \eqref{eq:FundamentalIneq} with $\theta = {\beta}_{n,M}$, the second inequality follows from the definition of $\beta_{n,M}$ (see Equation \eqref{eq:ModifiedObjective}) and the third inequality follows from inequality \eqref{eq:FundamentalIneq} with $\theta=\hat{\beta}_{n,M}$. Adding the sample average of $\{Y_i^2:\,1\le i\le n\}$ on both sides, we get for all $M\in\mathcal{M}_p(p)$,
\begin{equation}\label{eq:MainIneqLasso}
\hat{R}_n\left(\beta_{n,M}; M\right) \le \hat{R}_n\left(\hat{\beta}_{n,M}; M\right) + 2\mathcal{D}_{n}^{\Gamma}\left[\norm{\hat{\beta}_{n,M}}_1 + \norm{\beta_{n,M}}_1\right] + \mathcal{D}_{n}^{\Sigma}\left[\norm{\hat{\beta}_{n,M}}_1^2 + \norm{\beta_{n,M}}_1^2\right].
\end{equation}
This implies the first result \eqref{eq:LassoFinite}. To prove the second result \eqref{eq:LassoAsym}, note that
\begin{align*}
\left|\left(\frac{\mathcal{D}_{n}^{\Gamma} + \mathcal{D}_{n}^{\Sigma}\norm{\hat\beta_M}_1}{\mathcal{D}_{n}^{\Gamma} + \mathcal{D}_{n}^{\Sigma}\norm{\beta_{0,M}}_1}\right)^2 - 1\right|\le
\left|\frac{\mathcal{D}_{n}^{\Gamma} + \mathcal{D}_{n}^{\Sigma}\norm{\hat\beta_M}_1}{\mathcal{D}_{n}^{\Gamma} + \mathcal{D}_{n}^{\Sigma}\norm{\beta_{0,M}}_1} - 1\right|^2 + 2\left|\frac{\mathcal{D}_{n}^{\Gamma} + \mathcal{D}_{n}^{\Sigma}\norm{\hat\beta_M}_1}{\mathcal{D}_{n}^{\Gamma} + \mathcal{D}_{n}^{\Sigma}\norm{\beta_{0,M}}_1} - 1\right|,
\end{align*}
which converges to zero under assumption \ref{eq:UniformConsis}, following the proof of Theorem \ref{thm:PoSIX}. This implies that the error
\[
\left[2\mathcal{D}_{n}^{\Gamma}\norm{\hat\beta_{n,M}}_1 + \mathcal{D}_{n}^{\Sigma}\norm{\hat\beta_{n,M}}_1^2\right] - \left[2\mathcal{D}_{n}^{\Gamma}\norm{\beta_{n,M}}_1 + \mathcal{D}_{n}^{\Sigma}\norm{\beta_{n,M}}_1^2\right],
\]
is of smaller order than each of the terms uniformly in $M\in\mathcal{M}_p(k)$. The second result \eqref{eq:LassoAsym} then follows trivially by substituting the estimated parameters for the targets in inequality \eqref{eq:MainIneqLasso} and using the definition of $(C_{n}^{\Gamma}(\alpha), C_{n}^{\Sigma}(\alpha))$.

To prove the results with square-root lasso based regions, note that from inequality \eqref{eq:MainIneqLasso}
\begin{align*}
\hat{R}_n^{1/2}(\beta_{n,M}; M) &\le \hat{R}_n^{1/2}(\hat{\beta}_{n,M}; M) + \max\{\mathcal{D}_{n}^{\Gamma}, \mathcal{D}_{n}^{\Sigma}\}^{1/2}\left(1 + \norm{\hat{\beta}_{n,M}}_1\right)\\
&\quad+ \max\{\mathcal{D}_{n}^{\Gamma}, \mathcal{D}_{n}^{\Sigma}\}^{1/2}\left(1 + \norm{{\beta}_{n,M}}_1\right).
\end{align*}


\section{High-dimensional CLT and Bootstrap Consistency}\label{app:HDCLT}

Suppose $W_i, 1\le i\le n$ are independent random vectors in $\mathbb{R}^q$ with finite second moment. Let $G_i, 1\le i\le n$ be independent Gaussian random vectors in $\mathbb{R}^q$ with mean zero satisfying
\[
\mathbb{E}\left[G_iG_i^{\top}\right] = \mathbb{E}\left[W_iW_i^{\top}\right]\quad\mbox{for all}\quad 1\le i\le n.
\]
Set
\[
S_{n}^W := \frac{1}{\sqrt{n}}\sum_{i=1}^n \left\{W_i - \mathbb{E}\left[W_i\right]\right\}\quad\mbox{and}\quad S_n^G := \frac{1}{\sqrt{n}}\sum_{i=1}^n G_i.
\]
Before deriving the exact rate under the assumption \eqref{eq:MarginalPhi} of Lemma \ref{lem:RateD1nD2n}, we prove that a central limit theorem for $S_n^{W}$ implies a CLT for $(\mathcal{D}_{n}^{\Gamma}, \mathcal{D}_{n}^{\Sigma})$. Observe that for any $t_1, t_2 \in \mathbb{R}^+\cup\{0\}$,
\begin{equation}\label{eq:RectangleRepresentation}
\begin{split}
\left\{\mathcal{D}_{n}^{\Gamma} \le t_1,\right.&\left.\mathcal{D}_{n}^{\Sigma} \le t_2\right\} \\
&= \left\{-t_1 \le \frac{1}{n}\sum_{i=1}^n \left\{X_i(j)Y_i - \mathbb{E}\left[X_i(j)Y_i\right]\right\} \le t_1\mbox{ for all }1\le j\le p\right\}\bigcap\\
&\quad\left\{-t_2 \le \frac{1}{n}\sum_{i=1}^n \left\{X_i(l)X_i(m) - \mathbb{E}\left[X_i(l)X_i(m)\right]\right\} \le t_2\mbox{ for all }1\le l\le m\le p\right\}.
\end{split}
\end{equation}
The right hand side here is a rectangle in terms of the vector $S_n^W$ with vectors $W_i$ containing
\begin{equation}\label{eq:DefinitionWi}
\left(X_i(j)Y_i, 1\le j\le p;\; X_i(l)X_i(m), 1\le l \le m\le p\right).
\end{equation}
Note that $W_i$'s are vectors in $\mathbb{R}^{q}$ with
\[
q = 2p + \frac{p(p-1)}{2}.
\]
Let $\mathcal{A}^r$ denote the set of all rectangles in $\mathbb{R}^q$, that is, $\mathcal{A}^r$ consists of all sets $A$ of the form
\[
A = \{z\in\mathbb{R}^q:\,a(j) \le z(j) \le b(j)\,\mbox{ for all }\,1\le j\le q\},
\]
for some vectors $a, b\in\mathbb{R}^q$. Define
\begin{equation}\label{eq:MaxThirdMoment}
L_{n,q} := \max_{1\le j\le q}\,\frac{1}{n}\sum_{i=1}^n \mathbb{E}\left[\left|W_i(j) - \mathbb{E}[W_i(j)]\right|^3\right].
\end{equation}
Finally, set for any class $\mathcal{A}$ of (Borel) sets in $\mathbb{R}^q$,
\[
\rho_n\left(\mathcal{A}\right) := \sup_{A\in\mathcal{A}}\left|\mathbb{P}\left(S_n^W \in A\right) - \mathbb{P}\left(S_n^G \in A\right)\right|.
\]
The following theorem proved in Section 6 of \cite{KuchAbhi17} provides a central limit theorem for $S_n^W$ over all rectangles. The proof there is based on Theorem 2.1 of \cite{Chern17}.
\begin{thm}\label{thm:MarginalPhiHDBEBound}
Suppose $W_1, \ldots, W_n$ are independent mean zero random vectors in $\mathbb{R}^q$ satisfying for some $\gamma, B, K_{n,q} > 0$,
\begin{align}
&&\min_{1\le j\le q}\,\frac{1}{n}\sum_{i=1}^n \mbox{Var}\left[W_i^2(j)\right] \ge B && \mbox{and} 
&&\max_{1\le i\le n}\max_{1\le j\le q}\,\norm{W_i(j)}_{\psi_{\gamma}} &\le K_{n,q}. \label{eq:MarginalPhiBeta}
\end{align}
Assume further that for some constant $K_2 > 0$ (depending only $B$),
\begin{equation}\label{eq:SampleSizeDimension}
\frac{1}{8K_2K_{n,q}}\left(\frac{n L_{n,q}}{\log q}\right)^{1/3} \ge \max\{1,2^{1/\gamma - 1}\} \left\{ (\log q)^{1/\gamma} +  \left(6/\gamma\right)^{1/\gamma} + 1  \right\}.
\end{equation}\noeqref{eq:SampleSizeDimension}
Then there exist constants $K_1 > 0$ depending only on $B$, and $C_{\gamma, B} > 0$ depending only on $B, \gamma$ such that
\begin{align*}
\rho_n\left(\mathcal{A}^{re}\right) &\le K_1\left(\frac{L_{n,q}^2\log^7q}{n}\right)^{1/6} + C_{\gamma, B}\frac{K_{n,q}^6\log q}{n}.
\end{align*}
\end{thm}
Based on~\eqref{eq:RectangleRepresentation}, it is clear that Theorem~\ref{thm:MarginalPhiHDBEBound} implies a CLT for $(\mathcal{D}_n^{\Gamma}, \mathcal{D}_n^{\Sigma})$. This does not require the observations to be identically distributed or equal expectations for the $W_i$ vectors defined in~\eqref{eq:DefinitionWi}.


\subsection{Bootstrap Consistency}

In this sub-section, we consider the consistency of multiplier bootstrap based on Section 4.1 of \cite{Chern17}. It is also possible to consider the empirical bootstrap in high-dimensions and prove its consistency based on the proof of Proposition 4.3 of \cite{Chern17}. We do not prove it here as the proof techniques are the same.

Let $e_1, e_2, \ldots, e_n$ be a sequence of independent standard normal random variables independent of $\mathcal{W}_n := \{W_1, \ldots, W_n\}$. Set
\[
\bar{W}_n := \frac{1}{n}\sum_{i=1}^n W_i \in\mathbb{R}^q,
\]
and consider the normalized sum
\[
S_n^{eW} := \frac{1}{\sqrt{n}}\sum_{i=1}^n e_i\left(W_i - \bar{W}_n\right).
\]
Note that
\[
S_n^{eW}\big|\mathcal{W}_n\sim N\left(0, \frac{1}{n}\sum_{i=1}^n \left(W_i - \bar{W}_n\right)\left(W_i - \bar{W}_n\right)^{\top}\right)\in\mathbb{R}^q.
\]
To prove consistency of multiplier bootstrap, we bound a quantity similar to $\rho_n\left(\mathcal{A}^{re}\right)$, defined as
\begin{equation}\label{eq:MultiplierBEBound}
\rho_n^{\mbox{MB}}\left(\mathcal{A}^{re}\right) := \sup_{A\in\mathcal{A}^{re}}\left|\mathbb{P}\left(S_n^{eW} \in A\big|\mathcal{W}_n\right) - \mathbb{P}\left(S_n^{G\star} \in A\right)\right|,
\end{equation}
where
\[
S_n^{G\star} ~\sim~ N\left(0, \frac{1}{n}\sum_{i=1}^n \mathbb{E}\left[\left(W_i - \bar{\mu}_n\right)\left(W_i - \bar{\mu}_n\right)^{\top}\right]\right),\quad\mbox{with}\quad \bar{\mu}_n := \mathbb{E}\left[\bar{W}_n\right] = \frac{1}{n}\sum_{i=1}^n \mathbb{E}\left[W_i\right].
\]
Note that $\mbox{Var}\left(S_n^W\right) \neq \mbox{Var}\left(S_n^{G\star}\right)$ unless $\mathbb{E}[W_1] = \mathbb{E}[W_2] = \cdots = \mathbb{E}[W_n]$. Define
\begin{align*}
\Delta_{n,q} &:= \norm{\frac{1}{n}\sum_{i=1}^n \left(W_i - \bar{W}_n\right)\left(W_i - \bar{W}_n\right)^{\top} - \frac{1}{n}\sum_{i=1}^n \mathbb{E}\left[\left(W_i - \bar{\mu}_n\right)\left(W_i - \bar{\mu}_n\right)^{\top}\right]}_{\infty}.
\end{align*}
Based on Theorem 4.1 and Remark 4.1 of \cite{Chern17}, we prove the following theorem under assumption \eqref{eq:MarginalPhiBeta}.
\begin{thm}\label{thm:BootstrapConsistency}
If $W_i, 1\le i\le n$ are independent mean zero random vectors, then under assumption \eqref{eq:MarginalPhiBeta},
\[
\mathbb{E}\left[\rho_n^{\mbox{MB}}\left(\mathcal{A}^{re}\right)\right] \le C\log^{2/3}q\left[A_{n,q}^{1/3}\left(\frac{\log q}{n}\right)^{1/6} + K_{n,q}^{2/3}\frac{(\log q\log n)^{\frac{2}{3\gamma}}}{n^{1/3}}\right],
\]
for some constant $C$ depending only on $\gamma, B$. Here
\[
A_{n,q} := \max_{1\le l\le m\le q}\frac{1}{n}\sum_{i=1}^n \mbox{Var}\left(W_i(l)W_i(m)\right).
\]
\end{thm}
\begin{proof}
As proved in Remark 4.1 of \cite{Chern17}, we have
\[
\rho_n^{\mbox{MB}}\left(\mathcal{A}^{re}\right) \le C\Delta_{n,q}^{1/3}\log^{2/3}q.
\]
So, to prove the result, all we need is to prove
\[
\mathbb{E}\left[\Delta_{n,q}^{1/3}\right] \le M_{\gamma}\left[A_{n,q}\sqrt{\frac{\log q}{n}} + K_{n,q}^2(\log q\log n)^{2/\gamma}n^{-1}\right]^{1/3},
\]
for some constant $M_{\gamma}$. This follows from Theorem 4.2 of \cite{KuchAbhi17}.
\end{proof}
\begin{rem}{ (Inconsistency under unknown unequal means)}\label{rem:Inconsistency}
Since $\mbox{Var}(S_n^W)$ and $\mbox{Var}\left(S_n^{G\star}\right)$ are not equal (in general), Theorem~\ref{thm:BootstrapConsistency} does \emph{not} prove that
\[
\sup_{A\in\mathcal{A}^{re}}\left|\mathbb{P}\left(S_n^{eW} \in A\big|\mathcal{W}_n\right) - \mathbb{P}\left(S_n^{G} \in A\right)\right| \to 0.
\]
It was proved in \cite{Chap2:Kuch18} that variance of an average of non-identically distributed random variables cannot be consistently estimated if the expectations are unknown and the same comment applies to the high-dimensional multiplier bootstrap. When $\mathbb{E}\left[W_i\right]$ are not all the same for all $1\le i\le n$, then the variance of $S_n^{W}$ cannot be consistently estimated and so the distribution of $S_n^W$ cannot be estimated consistently using bootstrap. However, Theorem~\ref{thm:BootstrapConsistency} implies conservative inference. Observe that
\[
\mbox{Var}\left(S_n^W\right) = \frac{1}{n}\sum_{i=1}^n \mbox{Var}\left(W_i\right) \preceq \frac{1}{n}\sum_{i=1}^n \mathbb{E}\left[\left(W_i - \mathbb{E}\left[\bar{W}_n\right]\right)\left(W_i - \mathbb{E}\left[\bar{W}_n\right]\right)^{\top}\right].
\]
Hence by Anderson's Lemma (Corollary 3 of \cite{Anderson55}), for all $A\in\mathcal{A}^{sre},$
\[
\mathbb{P}\left(S_n^{G\star} \in A\right) \le \mathbb{P}\left(S_n^G \in A\right).
\]
Here $\mathcal{A}^{sre}$ represents the set of all rectangles that are symmetric around zero. Thus, we get that
\[
\liminf_{n\to\infty}\inf_{A\in\mathcal{A}^{sre}}\left( \mathbb{P}\left(S_n^{G} \in A\right) - \mathbb{P}\left(S_n^{eW} \in A\big|\mathcal{W}_n\right)\right) \ge 0.
\]
Observe that the sets in \eqref{eq:RectangleRepresentation} are centrally convex symmetric sets and so, Anderson's Lemma applies. Therefore, the multiplier bootstrap provides an asymptotically conservative inference for $(\mathcal{D}_n^{\Gamma}, \mathcal{D}_n^{\Sigma})$, in general.
\end{rem}
\section{Rate Bounds on $\mathcal{D}_n^{\Gamma}$ and $\mathcal{D}_n^{\Sigma}$ under Dependence}\label{app:Dependent}
In this section, we derive rate of convergence of $\|\hat{\Omega}_n - \Omega_n\|_{\infty}$ under dependence. We first describe some classical notions of dependence that include well-known dependent processes as special cases. The description is essentially taken from \cite{Potscher97}. Let $\{\xi_t:\,t\in\mathbb{Z}\}$ be a stochastic process on some measure space. Let $\mathcal{F}_{m,n}$ (for $m < n$) be the $\sigma$-field generated by $\{\xi_i:\,m \le i\le n\}$ with possibility of $m = -\infty$ and $n = \infty$ included. Define
\begin{align*}
\alpha(j) &:= \sup_{k\in\mathcal{Z}}\sup\left\{|\mathbb{P}\left(A\cap B\right) - \mathbb{P}(A)\mathbb{P}(B)|:\,A\in\mathcal{F}_{-\infty, j}, B\in\mathcal{F}_{k+j, \infty}\right\},\\
\phi(j) &:= \sup_{k\in\mathcal{Z}}\sup\left\{|\mathbb{P}\left(B|A\right) - \mathbb{P}(B)|:\,A\in\mathcal{F}_{-\infty, j}, B\in\mathcal{F}_{k+j, \infty}, \mathbb{P}(A) > 0\right\}.
\end{align*}
If $\alpha(j)$ (or correspondingly  $\phi(j)$) converges to zero as $j$ approaches infinity then the process $\{\xi_t:\,t\in\mathbb{Z}\}$ is called $\alpha$-mixing (or correspondingly $\phi$-mixing). It is clearly seen that every $\phi$-mixing process is $\alpha$-mixing since for any event $A$ with $\mathbb{P}(A) > 0,$
\begin{align*}
|\mathbb{P}\left(A\cap B\right) - \mathbb{P}(A)\mathbb{P}(B)| \le \mathbb{P}(A)|\mathbb{P}\left(B|A\right) - \mathbb{P}(B)|.
\end{align*}
A process $\{\xi_t:\,t\in\mathbb{Z}\}$ is said to be $m$-dependent if $\alpha(j) = 0$ for all $j\ge m$. Evidently, $m$-dependent processes are $\phi$-mixing for any $m$ and so $\alpha$-mixing too. One very useful feature of $\alpha$-mixing processes is that measurable functions of finitely many elements of the process themselves $\alpha$-mixing.

The dependence notion used in this section is the one called functional dependence introduced by \cite{Wu05}. It is possible to derive the results under the classical dependence notions like $\alpha$-,$\rho$- mixing too, however, verifying the mixing assumptions can often be hard and many well-known processes do not satisfy them. See \cite{Wu05} for more details. It has also been shown that many econometric time series can be studied under the notion of functional dependence; see \cite{Wu10}, \cite{Liu13} and \cite{WuWu16}.

The dependence notion of \cite{Wu05} is written in terms of an input-output process that is easy to analyze in many settings. The process is defined as follows. Let $\{\varepsilon_i, \varepsilon_i':\,i\in\mathbb{Z}\}$ denote a sequence of independent and identically distributed random variables on some measurable space $(\mathcal{E}, \mathcal{B})$. Let the $q$-dimensional (stochastic) process $W_i$ has a causal representation as
\begin{equation}\label{eq:CausalProcessVector}
W_i = G_i(\ldots,\varepsilon_{i-1}, \varepsilon_i)\in\mathbb{R}^q,
\end{equation}
for some vector-valued function $G_i(\cdot) = (g_{i1}(\cdot), \ldots, g_{iq}(\cdot))$. By Wold representation theorem for stationary processes, this causal representation holds in many cases. Define the non-decreasing filtration
\begin{equation}\label{eq:Filtration}
\mathcal{F}_i := \sigma\left(\ldots, \varepsilon_{i-1}, \varepsilon_i\right).
\end{equation}
Using this filtration, we also use the notation $W_i = G_i(\mathcal{F}_i)$. To measure the strength of dependence, define for $r\ge 1$ and $1\le j\le q$, the \textbf{functional dependence measure}
\begin{equation}\label{eq:FunctionalDepMeasure}
\delta_{s,r,j} := \max_{1\le i\le n}\,\norm{W_i(j) - W_{i,s}(j)}_r,\quad\mbox{and}\quad \Delta_{m,r,j} := \sum_{s = m}^{\infty} \delta_{s,r,j},
\end{equation}
where
\begin{equation}\label{eq:ZikDef}
W_{i,s}(j) := g_{ij}(\mathcal{F}_{i,i-s})\quad\mbox{with}\quad \mathcal{F}_{i,i-s} := \sigma\left(\ldots,\varepsilon_{i-s-1}, \varepsilon_{i-s}', \varepsilon_{i-s+1}, \ldots, \varepsilon_{i-1}, \varepsilon_i\right).
\end{equation}
The $\sigma$-field $\mathcal{F}_{i,i-s}$ represents a coupled version of $\mathcal{F}_i$. The quantity $\delta_{s,r,j}$ measures the dependence using the distance in terms of $\norm{\cdot}_r$-norm between $g_{ij}(\mathcal{F}_i)$ and $g_{ij}(\mathcal{F}_{i,i-s}).$ In other words, it is quantifying the impact of changing input $\varepsilon_{i-s}$ on the output $g_{ij}(\mathcal{F}_i)$; see Definition 1 of \cite{Wu05}. The \textbf{dependence adjusted norm} for $j$-th coordinate is given by
\[
\norm{\{W(j)\}}_{r,\nu} := \sup_{m\ge 0}\, (m + 1)^{\nu}\Delta_{m,r,j},\quad \nu \ge 0.
\]
To summarize these measures for the vector-valued process, define
\[
\norm{\{W\}}_{r,\nu} := \max_{1\le j\le q}\,\norm{\{W(j)\}}_{r,\nu}\quad\mbox{and}\quad \norm{\{W\}}_{\psi_{\beta}, \nu} := \sup_{r\ge 2}\,r^{-1/\beta}\norm{\{W\}}_{r,\nu}.
\]
\begin{rem}\,(Independent Sequences)\label{rem:IndepFunctionalDependence}
Any notion of dependence should at least include independent random variables. It might be helpful to understand how independent random variables fits into this framework of dependence. For independent random vectors $W_i$, the causal representation reduces to
\[
W_i = G_i(\ldots, \varepsilon_{i-1}, \varepsilon_i) = G_i(\varepsilon_i)\in\mathbb{R}^{q}.
\]
It is not a function of any of the previous $\varepsilon_j, j < i$. This implies by the definition~\eqref{eq:ZikDef}  that
\[
W_{i,s} = \begin{cases}
G_i(\varepsilon_i) = W_i,&\mbox{if }s\ge 1,\\
G_i(\varepsilon_{i}') =: W_i',&\mbox{if }s = 0.
\end{cases}
\]
Here $W_i'$ represents an independent and identically distributed copy of $W_i$. Hence,
\[
\delta_{s,r,j} = \begin{cases}
0, &\mbox{if }s\ge 1,\\
\norm{W_i(j) - W_i'(j)}_r \le 2\norm{W_i(j)}_r,&\mbox{if }s = 0.
\end{cases}
\]
It is now clear that for any $\nu > 0$,
\[
\norm{\{W\}}_{r, \nu} = \sup_{m\ge 0}\,(m + 1)^{\nu}\Delta_{m,r} = \Delta_{0, r} \le 2\max_{1\le j\le q}\norm{W_i(j)}_r.
\]
Hence, if the independent sequence $W_i$ satisfies assumption~\eqref{eq:MarginalPhi}, then $\norm{\{W\}}_{\psi_{\beta}, \nu} < \infty$ for all $\nu > 0$, in particular for $\nu = \infty$. Therefore, independence corresponds to $\nu = \infty$. As $\nu$ decreases to zero, the random vectors become more and more dependent.
\end{rem}
Recall that
\[
\norm{\hat{\Omega}_n - \Omega_n}_{\infty} := \max_{1\le j, k\le p+1}\left|\frac{1}{n}\sum_{i=1}^n \left(Z_i(j)Z_i(k) - \mathbb{E}\left[Z_i(j)Z_i(k)\right]\right)\right|,
\]
which is a maximum of $(p+1)^2$ many averages. To prove bound on the quantity above, consider the following assumption:
\begin{description}
\item[\namedlabel{eq:Dependent}{(DEP)}] Assume that there exist $n$ vector-valued functions $G_i$ and an iid sequence $\{\varepsilon_i:\,i\in\mathbb{Z}\}$ such that
\[
Z_i := (X_i, Y_i) = G_i(\ldots, \varepsilon_{i-1}, \varepsilon_i)\in\mathbb{R}^{p+1}.
\]
Also, for some $\nu, \beta > 0$,
\[
\norm{\{Z\}}_{\psi_{\beta}, \nu} \le K_{n,p}\quad\mbox{and}\quad \max_{1\le i\le n}\max_{1\le j\le p+1}\,|\mathbb{E}\left[Z_i(j)\right]| \le K_{n,p}.
\]
\end{description}
Based on Remark~\ref{rem:Indep}, Assumption~\ref{eq:Dependent} is equivalent to the assumption of Lemma~\ref{lem:RateD1nD2n} for independent data. For independent random variables, the second part of Assumption~\ref{eq:Dependent} about the expectations follows from the $\psi_{\beta}$-bound assumption. The reason for this expectation bound in the assumption here is that the functional dependence measure $\delta_{s,r}$ does not have any information about the expectation since
\[
\norm{W_i(j) - W_{i,s}(j)}_r = \norm{\left(W_i(j) - \mathbb{E}\left[W_i(j)\right]\right) - \left(W_{i,s}(j) - \mathbb{E}\left[W_{i,s}(j)\right]\right)}_r.
\]
The coupled random variable $W_{i,s}$ has the same expectation as $W_i$. Since the quantities we need to bound involve product of random variables, such a bound on the expectations is needed for our analysis.

The following result proves a bound on $\|\hat{\Omega}_n - \Omega_n\|_{\infty}$ under assumption~\ref{eq:Dependent}. Define
\[
\Upsilon_{4,p} := \max_{1\le j\le p+1}\,\left(\norm{\{Z(j)\}}_{4,0} + \max_{1\le i\le n}|\mathbb{E}|Z_i(j)||\right)\norm{\{Z(j)\}}_{4,\nu}.
\]
\begin{thm}\label{thm:BoundDependence}
Fix $n, k\ge 1$ and let $t\ge 0$ be any real number. Then under assumption~\ref{eq:Dependent}, with probability at least $1 - 8e^{-t}$, 
\[
\|\hat{\Omega}_n - \Omega_n\|_{\infty} \le 2eB_{\nu}\sqrt{\frac{\Upsilon_{4,p}(t + \log(4p))}{n}} + C_{\beta}K_{n,p}^2\frac{(\log n)^{1/s(\beta/2)}\kappa_n(\nu)(t + \log(4p))^{1/T_1(s(\beta/2))}}{n},
\]
where $T_1(\lambda) = \min\{\lambda, 1\}$, $s(\lambda) = (1/2 + 1/\lambda)^{-1}$ and
\[
\kappa_n(\nu) = 2^{\nu}\times\begin{cases}
5(\nu - 1/2)^{-3}, &\mbox{if }\nu > 1/2,\\
2(\log_2n)^{5/2}, &\mbox{if }\nu = 1/2,\\
5(2n)^{1/2 - \nu}(1/2 - \nu)^{-3}, &\mbox{if }\nu < 1/2.
\end{cases}
\]
Here $B_{\nu}$ and $C_{\beta}$ are positive constants depending only on $\nu$ and $\beta$.
\end{thm}
\begin{proof}
The proof follows from Lemma B.4 and Theorem 5.1 (or Theorem B.1) of \cite{Uniform:Kuch18}.
\end{proof}
\begin{rem}{ (Rate of Convergence under Dependence)}
Theorem~\ref{thm:BoundDependence} readily implies bounds on $C_n^{\Sigma}(\alpha)$ and $C_n^{\Gamma}(\alpha)$ along with rate bounds on $\mathcal{D}_n^{\Gamma}$ and $\mathcal{D}_n^{\Sigma}$.
\end{rem}
\end{document}